\documentclass[preprint]{sigplanconf}

\usepackage{tikz} 
\usepackage{amssymb,amsmath,amsthm,amsfonts}

\newtheorem{theorem}{Theorem}
\newtheorem{lemma}[theorem]{Lemma}

\newtheorem{proposition}[theorem]{Proposition}

\newcounter{venchoRestateLemma}
\newcounter{venchoRestateLemma1}
\newcounter{venchoRestateLemma2}

\begin{document}

\setlength{\pdfpageheight}{\paperheight}
\setlength{\pdfpagewidth}{\paperwidth}

\conferenceinfo{LICS'16}{July 5--8, 2016, New York City, USA}
\copyrightyear{2016}
\copyrightdata{978-1-nnnn-nnnn-n/yy/mm}
\copyrightdoi{nnnnnnn.nnnnnnn}

% Uncomment the publication rights you want to use.
%\publicationrights{transferred}
%\publicationrights{licensed}     % this is the default
%\publicationrights{author-pays}

\title{On Recurrent Reachability for Continuous Linear Dynamical Systems}
% double-check template to see how authors with different institutions work if
% we decide IST to appear under my name.
\authorinfo{Ventsislav Chonev}
			{Institute of Science and Technology Austria}
			{vchonev@ist.ac.at}
\authorinfo{Jo{\"e}l Ouaknine\and James Worrell}
           {University of Oxford}
           {\{joel,jbw\}@cs.ox.ac.uk}

\maketitle

\newcommand{\bigoh}{\mathcal{O}}
\newcommand{\ein}[1]{2^{\bigoh(||#1||)}}
\newcommand{\pin}[1]{||#1||^{\bigoh(1)}}
\newcommand{\ra}{\mathbb{R}\cap\mathbb{A}}
\newcommand{\spa}{\mathit{span}}
\newcommand{\alg}{\mathbb{A}}
\newcommand{\re}{\mathbb{R}}
\newcommand{\nat}{\mathbb{N}}
\newcommand{\nats}{\mathbb{N}}
\newcommand{\rat}{\mathbb{Q}}
\newcommand{\rats}{\rat}
\newcommand{\zed}{\mathbb{Z}}
\newcommand{\lalg}{\alg}
\newcommand{\defn}{\stackrel{\mathrm{def}}{=}}
\newcommand{\thr}{\mathit{Th(\re)}}
\newcommand{\threxp}{\mathit{Th(\re_{\mathit{exp}})}}
\newcommand{\lang}{\mathcal{L}}
\newcommand{\lexp}{\mathcal{L}_{\mathit{exp}}}
\newcommand{\pspace}{\mathsf{PSPACE}}
\newcommand{\mul}{\mathit{mul}}

\renewcommand{\Im}{\mathrm{Im}}
\renewcommand{\Re}{\mathrm{Re}}

\begin{abstract}
  The continuous evolution of a wide variety of systems, including
  continous-time Markov chains and linear hybrid automata, can be
  described in terms of linear differential equations. In this paper
  we study the decision problem of whether the solution
  $\boldsymbol{x}(t)$ of a system of linear differential equations
  $d\boldsymbol{x}/dt=A\boldsymbol{x}$ reaches a target halfspace
  infinitely often. This recurrent reachability problem can
  equivalently be formulated as the following Infinite Zeros Problem:
  does a real-valued function $f:\re_{\geq 0}\rightarrow\re$
  satisfying a given linear differential equation have infinitely many
  zeros?  Our main decidability result is that if the differential
  equation has order at most $7$, then the Infinite Zeros Problem is
  decidable. On the other hand, we show that a decision procedure for
  the Infinite Zeros Problem at order $9$ (and above) would entail a
  major breakthrough in Diophantine Approximation, specifically an
  algorithm for computing the Lagrange constants of arbitrary real
  algebraic numbers to arbitrary precision.
\end{abstract}



\section{Introduction}

A simple type of continuous-time system is one that satisfies a linear
differential equation $\frac{d\boldsymbol{x}}{dt} = A\boldsymbol{x}$,
where $A$ is an $n\times n$ matrix of rational numbers and
$\boldsymbol{x}(t) \in \mathbb{R}^n$ gives the system state at time
$t$.  In particular, such differential equations describe the state
evolution of finite-state continuous-time Markov chains (via the
so-called rate equation) and the continuous evolution of linear hybrid
automata.

A fundamental reachability question in this context is whether
$\boldsymbol{x}(t)$ infinitely often reaches a target hyperplane
$\boldsymbol{v}^T\boldsymbol{x}=0$, where
$\boldsymbol{v}\in \mathbb{R}^n$ is the normal vector.  Such a
hyperplane could represent a transition guard in a hybrid automaton or
a linear constraint on the state probability distribution of a
continuous-time Markov chain (e.g., that the probability to be in a
given state is exactly one half).

A solution $\boldsymbol{x}(t)$ to the differential equation
$\frac{d\boldsymbol{x}}{dt} = A\boldsymbol{x}$ admits a
matrix-exponential representation
$\boldsymbol{x}(t)=e^{At}\boldsymbol{x}(0)$.  The problem of whether
$\boldsymbol{x}(t)$ reaches the hyperplane
$\boldsymbol{v}^T\boldsymbol{x}=0$ infinitely often then reduces to
whether $\boldsymbol{v}^T\boldsymbol{x}(t)=0$ for infinitely many
values of $t\geq 0$.  Now the function
$f:\mathbb{R}_{\geq 0}\rightarrow\mathbb{R}$ defined by
$f(t)=\boldsymbol{v}^T\boldsymbol{x}(t)$ can be written as an
exponential-polynomial 
$f(t)=\sum_{j=1}^k P_j(t)e^{\alpha_j t}$, where the $\alpha_j$ and the
coefficients of each polynomial $P_j$ are algebraic numbers.  Thus the
problem of reaching a hyperplane infinitely often reduces to the
\emph{Infinite Zeros Problem}: given an exponential polynomial $f$, decide
whether $f$ has infinitely many non-negative real zeros.  Note that
since $f$ is an analytic function on the whole real line it only has
finitely many zeros in a bounded interval.  Thus the Infinite Zeros
Problem is equivalent to asking whether the set of zeros of $f$ is
unbounded.

This paper is concerned with the decidability of the Infinite Zeros
Problem.  In order to formulate our main results, recall that
exponential polynomials can equivalently be characterised as the
solutions of ordinary differential equations 
\begin{gather} 
f^{(n)} + a_{n-1} f^{(n-1)} + \ldots + a_0 f = 0 \, ,
\label{eq:ode}
\end{gather}
with the coefficients $a_j$ and the initial conditions $f^{(j)}(0)$
being real algebraic numbers for $j\in\{0,\dots,n-1\}$.  We say that
$f$ has order $n$ if it satisfies a linear differential equation of
the form (\ref{eq:ode}).
 
Our main results concern both decision procedures and hardness results
for the Infinite Zeros Problem.  We show that the problem is decidable
for exponential polynomials of order at most $7$.  With regards to
hardness, we exhibit a reduction to show that decidability of the
Continuous Infinite Zeros Problem for instances of order at least $9$
would entail major advancements in the field of Diophantine
Approximation, namely the computability of the Lagrange constants of
arbitrary real algebraic numbers.

Let us expand on the significance of the above hardness result.
Essentially nothing is known about the Lagrange constant of any real
algebraic number of degree three or above.  For example, it has been a
longstanding open problem since the 1930s whether some real algebraic
number of degree at least three has strictly positive Lagrange
constant and, on the other hand, whether some such number has Lagrange
constant $0$ (see, e.g.,~\cite{Guy04}).  These questions are often
formulated in terms of the simple continued fraction expansion of a
real number $\alpha$, which has unbounded elements if and only if
$\alpha$ has Lagrange constant $0$.

The reader will notice that there is a gap between our decidability
and hardness results for exponential polynomials of order $8$.  We
claim decidability in this case but defer the details to a longer
version of this paper.

Another way to calibrate the difficulty of the Infinite Zeros Problem
for an exponential polynomial $f(t)=\sum_{j=1}^k p_j(t)e^{\alpha_j t}$
is in terms of the dimension of the $\mathbb{Q}$-vector space spanned
by $\{ \Im(\alpha_j) : j=1,\ldots,k\}$.  We show decidability in case
this space is one-dimensional and we observe that the above hardness
result already applies in the two-dimensional case.

\subsection{Related Work}
Closely related to the Infinite Zeros Problem is the problem of
whether an exponential polynomial has \emph{some} zero.  This problem
is considered in~\cite{BDJB10} under the name Continuous Skolem-Pisot
Problem.  The techniques considered in the present paper are relevant
to the latter problem, but significant extra difficulties arise in this
new setting since we cannot discount the behaviour of $f$ on some
bounded initial segment of the reals.  Our work on the Continuous
Skolem-Pisot Problem will be reported elsewhere.

There is a natural discrete analog of the Infinite Zeros Problem:
given a linear recurrence sequence, determine whether it has
infinitely many zero terms. The decidability of the latter problem
was established by Berstel and Mignotte~\cite{BerstelMignotte76}.  The
problem of deciding whether a given linear recurrence sequence has
some zero term is called Skolem's Problem.  This is a longstanding and
celebrated open problem which essentially asks to give an effective
proof of the Skolem-Mahler-Lech Theorem for linear recurrences; see,
e.g., the exposition of Tao~\cite[Section 3.9]{Tao08}.

Macintyre and Wilkie~\cite{macintyreWilkie} showed decidability of the
first-order theory of $\langle\re,+,\times,0,1,<,\exp\rangle$, the
real field with exponentiation, subject to Schanuel's Conjecture in
transcendence theory.  In this paper we are concerned with the complex
exponential function, and we do not use this result.  Moreover,
although we do make use of transcendence theory, all the results in this
paper are unconditional.

\section{Mathematical Background}

\subsection{General Form of a Solution}
\label{sec:forms}
We recall some facts about the general form of solutions of ordinary
linear differential equations.  
Consider a homogeneous linear differential equation 
\begin{gather} 
f^{(n)} + c_{n-1} f^{(n-1)} + \ldots + c_0 f = 0 
\label{eq:ode2}
\end{gather}
of order $n$.  
The \emph{characteristic polynomial} of (\ref{eq:ode2}) is
\[ \chi(x) := x^n + c_{n-1} x^{n-1} + \ldots + c_0 \, .\]
If $\lambda$ is a root of $\chi$ of multiplicity $m$, then the
function $f(t)=t^je^{\lambda t}$ satisfies (\ref{eq:ode2}) for
$j=0,1,\ldots,m-1$.  There are $n$ distinct linearly independent
solutions of (\ref{eq:ode2}) having this form, and these span the
space of all solutions.

Let the distinct roots of $\chi$ be $\lambda_1,\ldots,\lambda_k$, with
respective multiplicities $m_1,\ldots,m_k$. We refer to $\lambda_1,\ldots,\lambda_k$ 
as the \emph{characteristic roots} of the differential equation. We also refer to 
the characteristic roots of maximum real part as \emph{dominant}. Write
$\lambda_j=r_j+ia_j$ for real algebraic numbers $r_j,a_j$,
$j=1,\ldots,k$. It follows from the discussion above that, given real
algebraic initial values of $f(0),f'(0),\ldots,f^{(n-1)}(0)$, the
uniquely defined solution $f$ of (\ref{eq:ode2}) can be written in one
of the following three equivalent forms.
\begin{enumerate}
\item As an \emph{exponential polynomial}
\[ f(t) = \sum_{j=1}^k P_j(t)e^{\lambda_jt} \]
where each $P_j$ is a polynomial with (complex) algebraic
coefficients and degree at most $m_j - 1$.
\item As a function of the form
  \[ f(t) = \sum_{j=1}^k e^{r_jt}(P_j(t)\cos(a_jt) +
  Q_j(t)\sin(a_jt)) \]
  where the polynomials $P_j, Q_j$ have real algebraic coefficients
  and degrees at most $m_j-1$.
\item As a function of the form
\[ f(t) = \sum_{j=1}^k e^{r_jt}\sum_{l=0}^{m_l-1}b_{j,l}t^l\cos(a_jt + 
\varphi_{j,l}) \]
where $b_{j,l}$ is real algebraic and $e^{i\varphi_{j,l}}$ algebraic for each $j,l$.
\end{enumerate}
We refer the reader to~\cite[Theorem 7]{BDJB10} for details.

%\begin{proof}
%The conversion of the original matrix formulation to
%forms 1, 2 was shown in \cite[Theorem 7]{BDJB10}.
%Form 1 is obtained from the Jordan canonical
%form of the matrix. If $A = P^{-1}JP$ where 
%$J = \mathit{diag}(J_1,\dots,J_m)$ is a Jordan matrix,
%the matrix exponential is then given by 
%\[ e^{At} = P^{-1}e^{Jt}P = 
%P^{-1}\mathit{diag}\left(e^{J_1t},\dots,e^{J_mt}\right)P \]
%together with the well-known closed form for the exponential
%of a Jordan block:
%\[
%\left(e^{J_st}\right)_{j, l} = e^{t\lambda_s}\frac{t^{j-l}}{(j-l)!}
%\mbox{ for $j\geq l$}
%\]
%Carrying out the multiplication immediately yields form 1. To obtain
%form 2, it suffices to observe that $f$ is real-valued
%and then to systematically take real parts everywhere. 
%Finally, to obtain form 3, for each $j$ in form 2, group terms of
%$P'_j$ and $P''_j$ of matching degree:
%\[
%f(t) = \sum_{j=1}^k e^{r_jt} 
%\sum_{l=0}^{m_j-1} t^l(c_{j,l}\cos(a_jt) + d_{j,l}\sin(a_jt)).
%\]
%Then take $b_{j, l} \defn \sqrt{c_{j,l}^2 + d_{j,l}^2}$ and let
%$\varphi_{j,l}$ be the angle in $[0, 2\pi)$ with
%$\cos(\varphi_{j,l}) = c_{j,l}/b_{j,l}$ and 
%$\sin(\varphi_{j,l}) = -d_{j,l}/b_{j,l}$.
%\end{proof}

\subsection{Number-theoretic tools}
Throughout this paper we denote by $\alg$ the set of algebraic numbers.
Recall that a standard way to represent an algebraic number $\alpha$
is by its minimal polynomial $M$ and a numerical approximation of
sufficient accuracy to distinguish $\alpha$ from the other roots of
$M$~\cite[Section 4.2.1]{Coh93}.  Given two algebraic numbers
$\alpha$ and $\beta$ under this representation, the \emph{Field
  Membership Problem} is to determine whether
$\beta \in \mathbb{Q}(\alpha)$ and, if so, to return a polynomial $P$
with rational coefficients such that $\beta=P(\alpha)$.  This problem
can be decided using the LLL algorithm, see~\cite[Section
4.5.4]{Coh93}.

Given the characteristic polynomial $\chi$ of a linear differential
equation we can compute approximations to each of its roots
$\lambda_1,\ldots,\lambda_n$ to within an arbitrarily small additive
error~\cite{panApproximatingRoots}.  Moreover, by repeatedly using an
algorithm for the Field Membership Problem we can compute a primitive
element $\theta$ for the splitting field of $\chi$ and representations
of $\lambda_1,\ldots,\lambda_n$ as polynomials in $\theta$.  Thereby we can
determine maximal $\mathbb{Q}$-linearly independent subsets of
$\{ \Re(\lambda_j) : 1 \leq j \leq n \}$ and
$\{ \Im(\lambda_j) : 1 \leq j \leq n \}$.

%Let $\log$ denote the branch of the complex logarithm defined by
%$\log(re^{i\theta}) = \log(r) + i\theta$ for a positive real number
%$r$ and $-\pi\leq\theta<\pi$.  Recall that one can compute $\log z$
%and $e^{z}$ to within arbitrarily small additive error given a
%sufficiently precise approximation of $z$~\cite{Brent76}.

We now move to some techniques from Transcendental Number Theory on
which our results depend in a critical way. The following theorem was
originally proven in 1934 by A. Gelfond \cite{gelfonda,gelfondb}
and independently by T. Schneider \cite{schneidera,schneiderb},
settling Hilbert's seventh problem in the affirmative.
\begin{theorem}\label{thm: gs}
(Gelfond-Schneider) If $a$ and $b$ are algebraic numbers with $a\neq
  0,1$ and $b\not\in \rat$, then $a^b$ is transcendental.
\end{theorem}
The following lemma, proven in \cite{BDJB10}, is a useful consequence 
of the powerful Baker's Theorem~\cite[Theorem 3.1]{Baker75}:
\begin{lemma}\label{lem: twoCosBaker}
\cite[Lemma 13]{BDJB10} Let $a,b\in\ra$ 
be linearly independent over $\rats$ and let 
$\varphi_1,\varphi_2$ be logarithms of algebraic numbers,
that is, $e^{i\varphi_1},e^{i\varphi_2}\in\alg$. There exist
effective constants $C,N,T>0$ such that for all $t\geq T$, at least
one of $1-\cos(at+\varphi_1) > C/t^N$ and $1-\cos(bt+\varphi_2) > C/t^N$
holds.
\end{lemma}

Another necessary tool is a version of Kronecker's well-known Theorem
in Diophantine Approximation.
\begin{theorem} (Kronecker, appears in \cite{hardy}) 
Let $\lambda_1,\dots,\lambda_m$ and $x_1,\dots,x_m$ be real numbers.
Suppose that for all integers $u_1,\dots,u_m$ such that
$u_1\lambda_1+\dots+u_m\lambda_m\in\zed$, we also have
$u_1x_1 + \dots + u_mx_m\in\zed$, that is, all
integer relations among the $\lambda_j$ also hold among
the $x_j$ (modulo $\zed$). Then for all
$\epsilon > 0$, there exist $p\in\zed^m$
and $n\in\nat$ such that $|n \lambda_j - x_j - p_j | < \epsilon$ for
all $1\leq j \leq m$. In particular, if $1,\lambda_1,\dots,\lambda_m$
are linearly independent over $\zed$, then there exist
such $n\in\nat$ and $p\in\zed^m$ for all $x\in\re^m$ and $\epsilon>0$.
\end{theorem}
A direct consequence is the following:
\begin{lemma}\label{lem: density}
Let $a_1,\dots,a_m\in\ra$ be linearly independent over $\rats$ and let
$\varphi_1,\dots,\varphi_m\in\re$.
Write $x \bmod 2\pi$ to denote $\min_{k\in\zed} |x + 2k\pi|$ for any $x\in\re$. 
Then the image of the mapping
$h(t) : \re_{\geq 0} \rightarrow [0,2\pi)^m$ given by 
\[ h(t) = ((a_1t + \varphi_1)\bmod 2\pi, \dots, (a_mt + \varphi_m)\bmod 2\pi) \]
is dense in $[0, 2\pi)^m$. Moreover, the set 
\[ \{ h(t)\, |\, (a_1t + \varphi_1) \bmod 2\pi = 0 \} \]
is dense in $\{0\}\times[0,2\pi)^{m-1}$.
\end{lemma}
\begin{proof}
For the first part of the claim, note that the 
linear independence
of $1,a_1/2\pi,\dots,a_m/2\pi$ follows from the linear independence of $a_1,\dots,
a_m$ and the transcendence of $\pi$. Then by Kronecker's Theorem, the restriction
$\{h(t)\, |\, t\in\nat \}$ is dense in $[0,2\pi)^m$, so certainly the whole image
of $h(t)$ must also be dense in $[0,2\pi)^m$. For the second part, the trajectory 
$h(t)$ has zero first coordinate precisely
when $t = -\varphi_1/a_1 + 2n\pi$ for some $n\in\zed$, at which times the trajectory
is
\begin{align*} g(n) & \defn 
h\left(\frac{-\varphi_1}{a_1} + 2n\pi\right) \\
& = \{0\}\times \left(n\frac{2\pi a_j}{a_1} + 
\frac{a_1\varphi_j-\varphi_1a_j}{a_1} \bmod 2\pi\right)_{2\leq j \leq m}
\end{align*}
As before, we have that $\{1, 2\pi a_2 / a_1, \dots, 2\pi a_m / a_1\}$ are linearly independent
over $\rat$ from the linear independence of $a_1,\dots,a_m$ and the transcendence of
$\pi$, so applying Kronecker's Theorem to the last $m-1$ components of this discrete
trajectory yields the second part of the claim.
\end{proof}

\subsection{First-Order Theory of the Reals}
We denote by $\lang$ the first-order language 
$\langle\re,+,\times,0,1,<\rangle$.  Atomic formulas in this language
are of the form $P(x_1,\dots,x_n) = 0$ and $P(x_1,\dots,x_n) > 0$ for
$P\in\zed[x_1,\dots,x_n]$ a polynomial with integer coefficients. A
set $X\subseteq \re^n$ is \emph{definable} in $\lang$ if there exists
some $\lang$-formula $\phi(\bar{x})$ with free variables $\bar{x}$
which holds precisely for valuations in $X$.  Analogously, a
function is definable if its graph is a definable set.  

We denote by $\thr$ \emph{the first-order theory of the reals}, that
is, the set of all valid sentences in the language $\lang$.  It is
worth remarking that any real algebraic number is readily definable
within $\lang$ using its minimal polynomial and a rational
approximation to distinguish it from the other roots. Thus, we can
treat real algebraic numbers constants as built into the language and use
them freely in the construction of formulas.  A celebrated result due
to Tarski \cite{Tarski51} is that the first-order theory of the reals
admits quantifier elimination: that each formula $\phi_1(\bar{x})$ in
$\lang$ is equivalent to some effectively computable formula
$\phi_2(\bar{x})$ which uses no quantifiers.  This immediately entails
the decidability of $\thr$. It also follows that sets definable in
$\lang$ are precisely the semialgebraic sets.  Tarski's original
result had non-elementary complexity, but improvements followed,
culminating in the detailed analysis of Renegar \cite{Renegar}.

Decidability and geometrical properties of definable sets in the
first-order theory of the structure
$\lexp=\langle\re,+,\times,0,1,<,\exp\rangle$, the reals with
exponentiation, have been explored by a number of authors.  Most
notably, Wilkie~\cite{wilkieMC} showed that the theory is
\emph{o-minimal} and Macintyre and Wilkie~\cite{macintyreWilkie}
showed that if Schanuel's conjecture is true then the theory is
decidable.  We will not need the above two results in this paper,
however we use the following, which is very straightforward to
establish directly.

\begin{proposition}
  There is a procedure that, given a semi-algebraic set
  $S\subseteq \mathbb{R}^k$ and real algebraic numbers
  $a_1,\ldots,a_k$, returns an integer $T$ such that
  $\{ t \geq 0 : (e^{a_1t},\ldots,e^{a_kt}) \in S \}$ either contains
  the interval $(T,\infty)$ or is disjoint from $(T,\infty)$.  The
  procedure also decides which of these two eventualities is the case.
\label{prop:o-minimal}
\end{proposition}
\begin{proof}
  Consider a polynomial $P \in \mathbb{Z}[u_1,\ldots,u_k]$.  For
  suitably large $t$ the sign of $P(e^{a_1t},\ldots,e^{a_kt})$ is
  identical to the sign of the coefficient of the dominant term in the
  expansion of $P(e^{a_1t},\ldots,e^{a_kt})$ as an exponential
  polynomial.  It follows that the sign of
  $P(e^{a_1t},\ldots,e^{a_kt})$ is eventually constant.  It is
  moreover clear that one can effectively compute a threshold beyond
  which the sign of $P(e^{a_1t},\ldots,e^{a_kt})$ remains the same.
  Since the set $S$ is defined by a Boolean combination of
  inequalities $P(u_1,\ldots,u_k) \sim 0$, for
  $\mathop{\sim}\in\{<,=\}$, the proposition immediately follows.
\end{proof}

\subsection{Useful Results About Exponential Polynomials}

We restate two useful theorems 
due to Bell et al.~\cite{BDJB10}. 

\begin{theorem}\label{blondelcomplex}
\cite[Theorem 12]{BDJB10}
Exponential polynomials with no real dominant 
characteristic roots have infinitely many zeros.
\end{theorem}

\begin{theorem}\label{blondelubcs}
\cite[Theorem 15]{BDJB10}
Suppose we are given an exponential polynomial whose 
dominant characteristic roots are simple, 
at least four in number and have
imaginary parts linearly independent over $\rats$. Then
the existence of infinitely many zeros is decidable. 
%Moreover,
%if there are finitely many zeros, there exists an effective
%threshold $T$ such that all zeros are in $[0, T]$.
\end{theorem}

\section{One Linearly Independent Oscillation}
In this section we consider exponential polynomials
$f(t) = \sum_{j=1}^k P_j(t) e^{\lambda_j t}$ under the assumption that
the span of $\{ \Im(\lambda_j) : j=1,\ldots,k\}$ is a one-dimensional
$\mathbb{Q}$-vector space.  In this case we can use fundamental
geometric properties of semi-algebraic sets to decide whether or not
$f$ has finitely many zeros. 
%and, if so, to compute an interval $[0,T]$ that contains all zeros of $f$.

\begin{theorem}\label{thm: mtargument}
  Let $f(t) = \sum_{j=1}^k P_j(t) e^{\lambda_j t}$ be an 
  exponential polynomial such that the span of
  $\{ \Im(\lambda_j) : j=1,\ldots,k\}$ is a one-dimensional
  $\mathbb{Q}$-vector space.  Then the existence of infinitely many
  zeros of $f$ is decidable. 
%and, if there are only finitely many
%  zeros, then there exists a computable bound $T$ such that all zeros
%  of $f$ lie in the interval $[0,T]$.
\end{theorem}
\begin{proof}
Write $\lambda_j=a_j+ib_j$, where $a_j,b_j$ are real algebraic numbers
for $j=1,\ldots,k$.  By assumption there is a single real algebraic
number $b$ such that each $b_j$ is an integer multiple of $b$.
Recall that for each integer $n$, both $\cos(nbt)$ and $\sin(nbt)$ can
be written as polynomials in $\sin(bt)$ and $\cos(bt)$ with integer
coefficients.  Using this fact we can write $f$ in the form
\[ f(t) = Q(t,e^{a_1t},\ldots,e^{a_kt},\cos(bt),\sin(bt)) \, , \]
for some multivariate polynomial $Q$ with algebraic coefficients.

Now consider the semi-algebraic set
\[ E:=\left \{ (\boldsymbol{u},s) \in \mathbb{R}^{k+2} :
  \textstyle Q\left(u_0,\ldots,u_k,\frac{1-s^2}{1+s^2},
    \frac{2s}{1+s^2}\right)=0 \right \} \, . \]
Recall that
$\left\{ \left(\frac{1-s^2}{1+s^2},\frac{2s}{1+s^2}\right) : s \in
  \mathbb{R} \right\}$
comprises all points in the unit circle in $\mathbb{R}^2$ except
$(-1,0)$.  Indeed, given $\theta \in (-\pi,\pi)$, setting
$s:=\tan(\theta/2)$ we have
$\cos(\theta)=\frac{1-s^2}{1+s^2}$ and $\sin(\theta)=\frac{2s}{1+s^2}$.
It follows that $f(t)=0$ and $\cos(bt)\neq -1$ imply that
$(t,e^{a_1t},\ldots,e^{a_kt},\tan(bt/2)) \in E$.

By the Cell Decomposition Theorem for semi-algebraic
sets \cite{marker}, there are semi-algebraic sets
$C_1,\ldots,C_m \subseteq \mathbb{R}^{k+2}$,
$D_1,\ldots,D_m \subseteq \mathbb{R}^{k+1}$, and
continuous semi-algebraic functions
$\xi_j,\xi^{(1)}_{j},\xi^{(2)}_{j} : D_j \rightarrow \mathbb{R}$
such that $E$ can be
written as a disjoint union $E=C_1\cup\ldots\cup C_m$, where either
\begin{eqnarray} C_j = \{ (\boldsymbol{u},s) \in \mathbb{R}^{k+2} :\boldsymbol{u}
\in D_j \wedge s=\xi_{j}(\boldsymbol{u}) \} \label{eq:zero-dim}
\end{eqnarray}
or
\begin{eqnarray}
C_j = \{ (\boldsymbol{u},s) \in \mathbb{R}^{k+2} :\boldsymbol{u} \in
D_j \wedge \xi^{(1)}_{j}(\boldsymbol{u}) <s< \xi^{(2)}_{j}(\boldsymbol{u})\} 
\label{eq:one-dim}
\end{eqnarray} 
Moreover such a decomposition is computable from $E$.  Clearly then
\[
\{t \in \mathbb{R} : f(t) = 0 \} \subseteq 
   \bigcup_{j=1}^m \{t \in \mathbb{R} : (t,e^{a_1t},\ldots,e^{a_kt}) \in D_j \} \cup Z \, ,
\]
where $Z:=\{  t \in \mathbb{R} : \cos(bt)=-1\}$.

The restriction of the exponential polynomial $f$ to $Z$ is given by
$f(t)=Q(t,e^{a_1t},\ldots,e^{a_kt},-1,0)$.  Since this expression is a
linear combination of terms of the form $t^je^{rt}$ for real algebraic
$r$, for sufficiently large $t$ the sign of $f(t)$ is determined by
the sign of the coefficient of the dominant term.  Thus $f$ is either
identically zero on $Z$ (in which case $f$ has infinitely many zeros)
or there exists some threshold $T$ such that all zeros of $f$ in $Z$
lie in the interval $[0,T]$.

We now consider zeros of $f$ that do not lie in $Z$.  There are two
cases.  First suppose that each set
$\{ t \in \mathbb{R} : (t,e^{a_1t},\ldots,e^{a_kt}) \in D_j \}$ is
bounded for $j=1,\ldots,m$.  In this situation, using
Proposition~\ref{prop:o-minimal}, we can obtain an upper bound $T$
such that if $f(t)=0$ then $t < T$.  On the other hand, if some set
$\{ t \in \mathbb{R}_{\geq 0} : (t,e^{a_1t},\ldots,e^{a_kt}) \in D_j \}$
is unbounded then, by Proposition~\ref{prop:o-minimal}, it contains an
infinite interval $(T,\infty)$.  We claim that in this case $f$ must
have infinitely many zeros $t\geq 0$.  We first give the argument in the
case $C_j$ satisfies (\ref{eq:zero-dim}).  

Define $\eta_j(t)=\xi_{j}(t,e^{a_1t},\ldots,e^{a_kt})$ for
$t \in (T,\infty)$.  Then for $t\in (T,\infty)\setminus Z$, 
\begin{eqnarray*}
 f(t)=0 &  \Longleftarrow & (t,e^{a_1t},\ldots,e^{a_kt},\tan(bt/2)) \in C_j\\
      &  \Longleftrightarrow  & (t,e^{a_1t},\ldots,e^{a_kt}) \in D_j \wedge \eta_j(t)=\tan(bt/2).
\end{eqnarray*}
In other words, $f$ has a zero at each point
$t \in (T,\infty) \setminus Z$ at which the graph of $\eta_j$
intersects the graph of $\tan(bt/2)$.  Since $\eta_j$ is continuous
there are clearly infinitely many such intersection points, see
Figure~\ref{fig:intersection}.  

The case when $C_j$ satisifes (\ref{eq:one-dim}) is handled similarly. In fact, this case
cannot arise at all, since by the above argument, if $C_j$ satisfies (\ref{eq:one-dim}),
then $f$ has a non-trivial interval of zeros. This is impossible, since $f$ is analytic,
and hence has only isolated zeros. This completes the proof.
\end{proof}
\begin{figure}
\begin{center}
\begin{tikzpicture}
  \draw[->] (0,0) -- (8.2,0) node[right] {$t$};
  \draw[<->] (0,-1.5) -- (0,1.5) ;
  \draw (0,0.75) node[left] {$s$};
 \draw plot[domain=0:0.8] (\x,{tan(1.2 * \x r)});
  \draw[dashed] (1,-1.5) -- (1,1.5);
\draw plot[domain=-0.8:0.8] (\x+2,{tan(1.2 * \x r)});
  \draw[dashed] (3,-1.5) -- (3,1.5);
\draw plot[domain=-0.8:0.8] (\x+4,{tan(1.2 * \x r)});
  \draw[dashed] (5,-1.5) -- (5,1.5);
\draw plot[domain=-0.8:0.8] (\x+6,{tan(1.2 * \x r)});
  \draw (3.5,0) node[below] {$T$};
  \draw (3.5,0.4) sin (8,0.1);
  \draw[dashed] (7,-1.5) -- (7,1.5);
 \draw plot[domain=-0.8:0] (\x+8,{tan(1.2 * \x r)});
  \draw (3.3,0.2) node[above] {$\eta_j$};
\end{tikzpicture}
\end{center}
\caption{Intersection points of $\eta_j(t)$ and $\tan(bt/2)$.}
\label{fig:intersection}
\end{figure}
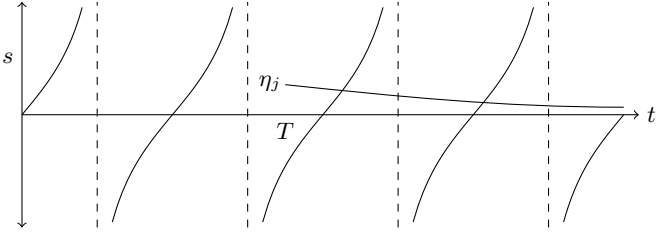

\section{Decidability Results up to Order 7}

We now shift our attention to instances of the Infinite Zeros Problem
of low order.  
%Given an exponential polynomial $f(t)$, we will once
%again be interested in two questions: does $f$ have infinitely many
%zeros, and if not, can we derive a bound $T$ such that all zeros of
%$f$ lie in the interval $[0,T]$? 
In particular, for exponential
polynomials corresponding to differential equations of order at most
7, we establish decidability of the
Infinite Zeros Problem.
%and a reduction from the Unbounded Skolem
%Problem to the Bounded Skolem Problem.

\begin{theorem}\label{thm: main}
  The Infinite Zeros Problem is decidable for differential
  equations of order at most $7$.
\end{theorem}

\begin{proof}
Suppose we are given an exponential polynomial $f$ of order at most $7$.
Sort the characteristic roots according to
their real parts, and let $r_j$ denote throughout the $j$-th largest
real part of a characteristic root. We will
refer to the characteristic roots of maximum real part as
the \emph{dominant characteristic roots}. 
Let also $\mul(\lambda)$ denote the multiplicity 
of $\lambda$ as a root of the characteristic polynomial of 
$f$.
%the given ODE. 

We will now perform a case analysis on the number of dominant 
characteristic roots. By Theorem \ref{blondelcomplex}, it is sufficient
to confine our attention to exponential polynomials with an odd number 
of dominant characteristic roots. Throughout, we rely on known 
general forms of solutions to ordinary linear differential equations,
outlined in Section \ref{sec:forms}.

\emph{Case I.} Suppose first that there is only one dominant, 
necessarily real, root $r$. Then if we divide $f$ by $e^{rt}$,
we have:
\[ \frac{f(t)}{e^{rt}} = P_1(t) + \bigoh\left(e^{(r_2-r)t}\right), \]
as the contribution of the non-dominant roots shrinks exponentially,
relative to that of the dominant root. Thus, for large $t\geq 0$, the 
sign of $f(t)$ matches the sign of the leading coefficient 
of $P_1(t)$, so $f$ cannot have infinitely many zeros. 
%Further, a bound $T$ on the zeros of $f(t)$ can be found easily from the
%description of $f(t)$.

\emph{Case II.} We now move to the case of three dominant characteristic roots:
$r$ and $r \pm ia$, so that
\[ \frac{f(t)}{e^{rt}} = P_1(t) + P_2(t)\cos(at) + P_3(t)\sin(at) + 
\bigoh\left(e^{(r_2-r)t}\right), \]
where $P_1,P_2,P_3\in(\ra)[x]$ have degrees $d_1\defn\deg(P_1) \leq \mul(r)-1$ 
and $d_2\defn\deg(P_2) = \deg(P_3) \leq \mul(r\pm ai)$.

\emph{Case IIa.} Suppose $d_1 > d_2$.
Now, it is easy to see that for large $t$ the sign of
$f(t)$ matches the sign of the leading coefficient $p_1$ of
$P_1$:
\[ 
\frac{f(t)}{e^{rt}t^{d_1}} = p_1 + \bigoh(1/t) + 
\bigoh\left(e^{(r_2-r)t}\right),
\]
so clearly some bound $T$ exists such that $t > T\Rightarrow f(t)\neq0$.
Similarly, if $d_2 > d_1$,
then $f(t)$ clearly has infinitely many zeros. Indeed, if $p_2, p_3$ are 
the leading coefficients of $P_2,P_3$, respectively, then we have:
\begin{align*}
\frac{f(t)}{e^{rt}t^{d_2}} & = p_2\cos(at) + p_3\sin(at) 
+ \bigoh(1/t) + \bigoh\left(e^{(r_2-r)t}\right) \\
& = \frac{\cos(at + \varphi)}{\sqrt{p_2^2+p_3^2}}
+ \bigoh(1/t) + \bigoh\left(e^{(r_2-r)t}\right) \\
\end{align*}
where $\varphi\in[0,2\pi)$ with $\tan(\varphi) = -p_3/p_2$,
so $f$ is infinitely often positive and infinitely often negative.

Thus, we can now assume $d_1 = d_2$. Notice that since the order of
our exponential polynomial is no greater than $7$, we must have 
$d_1 = d_2 \leq 2$.

\emph{Case IIb.} Suppose that $d_1 = d_2 = 2$. Then
our function is of the form
\[ 
\frac{f(t)}{e^{rt}} = t(A\cos(a t+\varphi_1) + B) + (C\cos(a t + \varphi_2) + D) + e^{(r_2-r)t}F,
\]
for constants $A,B,C,D,F,a\in\ra$ with $a>0$ and $e^{i\varphi_1},e^{i\varphi_2}\in\alg$.
In this case, Theorem \ref{thm: main} follows from Lemma \ref{lem: repOsc} in Section
\ref{sec:repOsc}.

\emph{Case IIc.} Suppose that $d_1 = d_2 = 1$, so that
\[ 
\frac{f(t)}{e^{rt}} = A_1\cos(at+\varphi_1) + A_2 + e^{(r_2-r)t}F_1(t),
\]
where $A_1,A_2,a\in\ra$, $a>0$, $e^{i\varphi_1}\in\alg$ and 
$F_1$ is an exponential polynomial with dominant characteristic root
whose real part is $0$. Consider first the magnitudes of $A_1$ and $A_2$. 
If $|A_1|>|A_2|$, then the term $A_1\cos(at+\varphi_1)$ makes $f$ change 
sign infinitely often, so $f$ must have infinitely many zeros. 
On the other hand, if $|A_1|<|A_2|$, then $f$ is clearly ultimately positive
or ultimately negative, depending on the sign of $A_2$.
%with an effective threshold beyond which $f(t)\neq 0$.
The remaining case is that $|A_1| = |A_2|$. Dividing $f$ by $A_2$, replacing
$\varphi_1$ by $\varphi_1+\pi$ if needed and scaling constants by $A_2$ as necessary, 
we can assume the function has the form:
\[ \frac{f(t)}{e^{rt}} = 1-\cos(at+\varphi_1) + e^{(r_2-r)t}F_1(t). \]
We now enumerate the possibilities for the dominant characteristic roots
of the exponential polynomial $F_1$, that is,
the characteristic roots of $f$ with second-largest real part. 
Since $f$ has order at most $7$, 
there are the following cases to consider:
\begin{itemize}
\item $F_1$ has four simple, necessarily complex, dominant roots, so that
\begin{align*}
\frac{f(t)}{e^{rt}} = \; & 1-\cos(at + \varphi_1) \\
& + e^{(r_2-r)t}(B\cos(bt + \varphi_2) + C\cos(ct+\varphi_3)),
\end{align*}
where $B,C,b,c\in\ra$ with $b,c>0$ and $e^{i\varphi_2},e^{i\varphi_3}\in\alg$.
In this case, Theorem \ref{thm: main} follows from Lemma \ref{lem: oneOscTwo}
in Section \ref{sec:oneOsc}.
\item $F_1$ has some subset of one real and two complex numbers as
dominant roots, all simple, so that
\begin{align*}
\frac{f(t)}{e^{rt}} = \; & 1-\cos(at + \varphi_1) \\
& + e^{(r_2-r)t}(B\cos(bt + \varphi_2) + C)
+ e^{(r_3-r)t}F_2(t),
\end{align*}
where $B,C,b\in\ra$, $b>0$, $e^{i\varphi_2}\in\alg$ and
$F_2$ is an exponential polynomial with dominant characteristic
root whose real part is $0$. 
In this case, Theorem \ref{thm: main} follows from Lemma \ref{lem: layered}
in Section \ref{sec:oneOsc}.
\item $F_1$ has a repeated real and possibly two simple complex dominant
roots, so that
\begin{align*}
 \frac{f(t)}{e^{rt}} = \; &
 1-\cos(at + \varphi_1) \\
& + e^{(r_2-r)t}(B\cos(bt + \varphi_2) + P(t))
+ e^{(r_3-r)t}F_2(t),
\end{align*}
where $B,b\in\ra$, $b>0$, $e^{i\varphi_2}\in\alg$, and
$P\in(\ra)[x]$ is non-constant. Now, if the leading coefficient of
$P$ is negative, then $f$ will be infinitely often negative
(consider large times $t$ such that $\cos(at+\varphi_1)=1$) and
infinitely often positive (consider large times $t$ such that 
$\cos(at+\varphi_1)=0$),
so $f$ must have infinitely many zeros. On the other hand,
if the leading coefficient of $P$ is positive, then it is
easy to see that $f$ is ultimately positive. 
%with an effective threshold.
\item $F_1$ has a repeated pair of complex roots, so that
\begin{align*}
\frac{f(t)}{e^{rt}} = \; & 1-\cos(at + \varphi_1) \\
& + e^{(r_2-r)t}(Bt\cos(bt + \varphi_2) + C\cos(bt+\varphi_3)),
\end{align*}
where $B,C,b\in\ra$, $b>0$ and $e^{i\varphi_2}, e^{i\varphi_3}\in\alg$.
In this case, Theorem \ref{thm: main} follows from Lemma \ref{lem: oneOscOneRep}
in Section \ref{sec:oneOsc}.
\end{itemize}

\emph{Case III.} We now consider the case of five dominant characteristic roots.
Let these be $r$, $r \pm ai$ and $r \pm bi$. If $r \pm ai$ are
repeated, i.e., $\mul(r\pm ai) \geq 2$, then we must have
$\mul(r) = \mul(r\pm bi) = 1$, since otherwise the order of our
exponential polynomial exceeds $7$. Then by an argument analogous to
\emph{Case IIa} above, $f$ must have infinitely many zeros.
The situation is symmetric when $\mul(r\pm bi)\geq 2$. Similarly,
if $\mul(r)\geq 2$, then $\mul(r\pm ai) = \mul(r\pm bi) = 1$, since
otherwise the instance exceeds order $7$. Then by the same argument
as in \emph{Case IIa}, $f$ is ultimately positive or ultimately
negative. 
%with an effectively computable threshold $T$. 
Thus, we may assume that all the dominant roots are simple, so the 
exponential polynomial is of the form:
\[ 
\frac{f(t)}{e^{rt}} = A\cos(at + \varphi_1) + B\cos(bt+\varphi_2) + C
+ e^{(r_2-r)t}F(t),
\]
where $A,B,C,a,b\in\ra$, $a,b>0$, $e^{i\varphi_1},e^{i\varphi_2}\in\alg$ and
$F$ is an exponential polynomial of order at most $2$ whose 
dominant characteristic roots have real part equal to $0$. 
In this case, Theorem \ref{thm: main} follows from Lemma
\ref{lem: twoOsc} in Section \ref{sec:twoOsc}.

\emph{Case IV.} Finally, suppose there are seven dominant characteristic
roots: $r$, $r\pm ai$, $r\pm bi$ and $r\pm ci$. Since we are limiting
ourselves to instances of order $7$, these roots must all be simple,
and there can be no other characteristic roots. Thus, the
exponential polynomial has the form
\[ 
\frac{f(t)}{e^{rt}} = A\cos(at + \varphi_1) + B\cos(bt+\varphi_2) + 
C\cos(ct+\varphi_3) + D,
\]
with $A,B,C,D,a,b,c\in\ra$, $a,b,c>0$ and $e^{i\varphi_1},\dots,e^{i\varphi_3}
\in\alg$. 
In this case, Theorem \ref{thm: main} follows from Lemma \ref{lem: threeOsc}
in Section~\ref{sec:threeOsc}.
\end{proof}

In the remainder of this section, we provide the technical lemmas invoked
throughout the proof of Theorem~\ref{thm: main}.
%\begin{corollary}
%\label{corl:main}
%  For differential equations of order at most $7$, the 
%  Infinite Zeros Problem is decidable.
%\end{corollary}
\subsection{One dominant oscillation}\label{sec:oneOsc}

\begin{lemma}\label{lem: taylor}
Let $A, B, a, b, r \in \ra$ where $a, b, r > 0$. Let
$\varphi_1,\varphi_2\in\mathbb{R}$ be such that 
$e^{i\varphi_1}, e^{i\varphi_2}\in\lalg$. Suppose
also that $a, b$ are linearly dependent over $\rats$ and that
whenever $1-\cos(at + \varphi_1) = 0$, it holds that 
$A\cos(bt + \varphi_2) + B > 0$. Define the function
\[ f(t) = 1-\cos(at + \varphi_1) + e^{-rt}(A\cos(bt + \varphi_2) + B). \]
Then $f(t) = \Omega(e^{-rt})$, that is, there exist effective constants
$T\geq 0$ and $c > 0$ such that for $t \geq T$, we have $f(t)\geq ce^{-rt}$.
\end{lemma}
\begin{proof}
The case of $A=0$ is easy: by the premise of the Lemma, we
have $B>0$ and then $f(t)\geq Be^{-rt}$ for all $t$.
Thus, assume $A\neq 0$ throughout.
Let the linear dependence between $a, b$ be given by $an_1 - bn_2 = 0$
for $n_1,n_2\in\mathbb{N}$ coprime and let
$\mathcal{C}$ be the equivalence class of $-\varphi_1/a$ 
modulo $2\pi/a$, that is,
\[
\mathcal{C} \defn \left\{ \frac{-\varphi_1 + 2k\pi}{a} \, \middle| \, k\in\mathbb{Z} \right\}.
\]
We will refer to $\mathcal{C}$ as the set of \emph{critical points}
throughout. 

It is clear that at critical points, we have $1-\cos(at +
\varphi_1)=0$. Moreover, the linear dependence of $a,b$
entails that for each fixed value of $(\cos(at),\sin(at))$,
there are only finitely many possible values for $(\cos(bt),\sin(bt))$.
Indeed, we have
\[ e^{ibt}\in\{\omega e^{iatn_1} \, | \, \omega\mbox{ an $n_2$-th root of unity} \}, \]
so in particular, for $t\in\mathcal{C}$, we have
\[
e^{ibt}\in\{\omega e^{-in_1\varphi_1}\, |\, \omega\mbox{ an $n_2$-th root of unity} \}.
\]
Thus, the possible values of $(\cos(bt), \sin(bt))$
for $t$ critical are algebraic and effectively computable.
Let $M \defn \min\{ A\cos(bt+\varphi_2)+B\,|\,t\in\mathcal{C}\}$.
By the premise of the Lemma, we have $M>0$.

Let $t_1,t_2,\dots,t_j,\dots$ be the non-negative critical points.
Note that by construction we have $|t_j - t_{j-1}| = 2\pi/a$.
For each $t_j$, define the \emph{critical region} to be the interval
$[t_j-\delta, t_j+\delta]$, where
\[ \delta\defn \frac{M}{2|A|b}. \]

Let $g(t)\defn A\cos(bt+\varphi_2)+B$ and notice that $g'(t)\leq |A|b$
everywhere. We first prove the claim for $t$ inside critical
regions: suppose $t$ lies in a critical region and let $j$ minimise 
$|t-t_j|\leq\delta$. Then by the Mean Value Theorem, we have 
\[ 
|g(t) - g(t_j)| \leq |t-t_j| |A|b \leq \delta|A|b = \frac{M}{2},
\]
so 
\[ g(t)\geq g(t_j) - \frac{M}{2} \geq \frac{M}{2}, \]
whence $f(t) \geq e^{-rt}g(t) \geq Me^{-rt}/2 = \Omega(e^{-rt})$.

Now suppose $t$ is outside all critical regions and let $j$
minimise $|t-t_j|$. Since the distance between critical 
points is $2\pi/a$ by construction, we have $a|t-t_j|\leq \pi$. 
Therefore,
\begin{align*}
1-\cos(at + \varphi_1) & = 1-\cos(at - at_j) \geq 
\frac{|a(t-t_j)|^2}{2} \\
& > \frac{(a\delta)^2}{2} = \frac{a^2M^2}{8|A|^2b^2}>0.
\end{align*}
Thus, there exists a computable constant $D>0$ such that 
$f(t)=1-\cos(at+\varphi_1) + e^{-rt}g(t)\geq D$
for all large enough $t$ outside critical regions.

Combining the two results, we have $f(t) = \Omega(e^{-rt})$
everywhere.
\end{proof}

\begin{lemma}\label{lem: layered}
Let $C, D, a, b, r_1, r_2$ be real algebraic numbers such that 
$a, b, r_1, r_2 > 0$ and $C, D$ are not both $0$.
Let also $\varphi_1,\varphi_2\in\mathbb{R}$ be
such that
$e^{i\varphi_1},e^{i\varphi_2}\in\lalg$.  
Define the exponential polynomial $f$ by
\begin{align*}
f(t) = \; & 1-\cos(at + \varphi_1) \\
& + e^{-r_1t}(C\cos(bt+\varphi_2) + D) + e^{-(r_1+r_2)t}F(t).
\end{align*}
Here $F$ is an exponential polynomial whose
dominant characteristic roots are purely imaginary. Suppose also
that $f$ has order at most $7$.
Then it is decidable whether $f$ has infinitely many
zeros. 
%Moreover, if $f(t)$ has only finitely many zeros, then there
%exists an effectively computable threshold $T$ such that all zeros of $f(t)$ are
%contained in $[0,T]$.
\end{lemma}
\begin{proof}
%Consider first the magnitudes of $A$ and $B$. If $|A|>|B|$, then 
%the term $A\cos(at+\varphi_1)$ makes $f(t)$ change sign infinitely
%often, so $f(t)$ must have infinitely many zeros. On the other 
%hand, if $|A|<|B|$, then $f(t)$ is clearly ultimately positive
%or ultimately negative, depending on the sign of $B$.
%Thus, we can assume $|A| = |B|$. Dividing $f(t)$ by $B$, replacing
%$\varphi_1$ by $\varphi_1+\pi$ if needed and scaling constants by $B$ as necessary, 
%we can assume the function has the form:
%\[ f(t) = 1-\cos(at+\varphi_1) + e^{-r_1t}(C\cos(bt+\varphi_2)+D) + e^{-(r_1+r_2)t}F(t). \]
Notice that the dominant term of $f$ is always non-negative, so the
function is positive for arbitrarily large $t$. Thus, $f(t)=0$ for
some $t$ if and only if $f(t)\leq 0$ for some $t$, and analogously,
$f$ has infinitely many zeros if and only if $f(t)\leq 0$
infinitely often. We can eliminate the 
case $|D|>|C|$, since then $f$ is clearly ultimately positive 
or oscillating, depending on the sign of $D$. 
Thus, we can assume $|D|\leq|C|$.

We now consider two cases, depending on whether $a/b \in \rats$.

\emph{Case I.}
Suppose first that $a, b$ are linearly independent over
$\mathbb{Q}$. By Lemma \ref{lem: density}, the trajectory
$(at+\varphi_1 \bmod 2\pi,bt + \varphi_2\bmod 2\pi)$ 
is dense in $[0,2\pi)^2$, and moreover the restriction
of this trajectory to $at + \varphi_1 \bmod 2\pi = 0$ is
dense in $\{0\}\times [0, 2\pi)$. 

If $|D| < |C|$, then we argue that $f$ is infinitely often 
negative, and hence has infinitely many zeros.
Indeed, $|D| < |C|$ entails the existence of a non-trivial interval 
$I\subseteq [0,2\pi)$ such that 
\[ t\bmod 2\pi\in I \Rightarrow C\cos(bt+\varphi_2) + D < 0. \]
What is more, we can in fact find $\epsilon>0$ and a subinterval 
$I' \subseteq I$ such that
\[ t\bmod 2\pi\in I' \Rightarrow C\cos(bt+\varphi_2)+D < -\epsilon. \] 
Thus, by density, 
$1 - \cos(at + \varphi_1)=0$ and $C\cos(bt+\varphi_2) + D < -\epsilon$
will infinitely often hold simultaneously. Then just take $t$ large enough to 
ensure, say, $|e^{-r_2t}F(t)| < \epsilon/2$ at these infinitely
many points, and the claim follows.

Thus, suppose now $|C|=|D|$. Replacing $\varphi_2$ by $\varphi_2+\pi$ 
if necessary, we can write the function as:
\begin{align*}
f(t) = \; & 1-\cos(at + \varphi_1) \\
& + De^{-r_1t}(1-\cos(bt+\varphi_2)) + e^{-(r_1+r_2)t}F(t).
\end{align*}
As $a,b$ are linearly independent, 
for all $t$ large enough, $1-\cos(at+\varphi_1)$ and $1-\cos(bt+\varphi_2)$
cannot simultaneously be `too small'. More precisely, by Lemma \ref{lem: twoCosBaker}, 
there exist effective constants $E, T, N>0$ such that for all $t\geq T$,
we have
\begin{align*}
1-\cos(at+\varphi_1) > E / t^N \mbox{ or }  
1-\cos(bt + \varphi_2) > E / t^N. \\
\end{align*}
Now, if $D<0$, it is easy
to show that $f$ has infinitely many zeros. Indeed, consider the
times $t$ where the dominant term $1-\cos(at + \varphi_1)$
vanishes. For all large enough such $t$, since $t^{-N}$ shrinks more 
slowly than $e^{-r_2t}$, we will have 
\begin{align*}
f(t) & = e^{-r_1t}D(1-\cos(bt + \varphi_2)) + e^{-(r_1+r_2)t}F(t)  \\
& < e^{-r_1t}( EDt^{-N} + e^{-r_2t}F(t) )  \\
& \leq e^{-r_1t}\frac{1}{2}EDt^{-N} \\
& < 0, \\
\end{align*}
so $f$ has infinitely many zeros.
Similarly, if $D>0$, we can show that $f$ is
ultimately positive.  Indeed, for all $t$ large enough, 
we have
\begin{align*}
f(t) & \geq e^{-r_1t}D(1-\cos(bt+\varphi_2)) + e^{-(r_1+r_2)t}F(t) \\
& > e^{-r_1t}DE t^{-N} + e^{-(r_1+r_2)t}F(t) \\
& > 0, \\
\end{align*}
or
\begin{align*}
f(t) & \geq 1-\cos(at+\varphi_1) + e^{-(r_1+r_2)t}F(t) \\
& > E t^{-N} + e^{-(r_1+r_2)t}F(t) \\
& > 0. \\
\end{align*}
Therefore, $f$ has only finitely many zeros.
%, all occurring up to some effective bound $T$.

\emph{Case II.} Now suppose $a,b$ are linearly dependent. By the premise of
the Lemma, the order of $F$ is at most $2$ (in fact, 
at most $1$ if $D\neq 0$). However, by Theorem \ref{thm: mtargument}, the claim follows immediately
for all cases in which the characteristic roots of $F$ are all real
or complex but with frequencies linearly dependent on $a$. Thus, the only remaining
case to consider is the function
\begin{align*} 
f(t) = \; & 1-\cos(at+\varphi_1) \\
& + e^{-r_1t}C\cos(bt+\varphi_2) + e^{-(r_1+r_2)t}H\cos(ct+\varphi_3), 
\end{align*}
where $H,c\in\ra$, $c > 0$ and $a/c\not\in\rats$.

As explained at the beginning of the proof of Lemma \ref{lem: taylor}, due to the linear
dependence of $a, b$ over $\rats$, when $1-\cos(at+\varphi_1)=0$, there are only finitely
many possibilities for the value of $C\cos(bt+\varphi_2)$, each algebraic, effectively
computable and occurring periodically. If at least one of these values is non-positive,
then by the linear independence of $a,c$ over $\rats$, we will simultaneously have 
$1-\cos(at+\varphi_1) = 0$, $C\cos(bt+\varphi_2)\leq 0$ and $H\cos(ct+\varphi_3) < 0$
infinitely often, which yields $f(t)<0$ infinitely often and entails the existence of
infinitely many zeros. On the other hand, if at the critical points $1-\cos(at+\varphi_1) = 0$
we always have $C\cos(bt+\varphi_2) > 0$, then by Lemma \ref{lem: taylor}, we have 
\[ 1-\cos(at+\varphi_1) + e^{-r_1t}C\cos(bt+\varphi_2) = \Omega(e^{-r_1t}), \]
whereas obviously 
\[ \left|e^{-(r_1+r_2)t}H\cos(ct+\varphi_3)\right| = \bigoh(e^{-(r_1+r_2)t}). \]
If follows that $f$ is ultimately positive and hence has only finitely many zeros.
%An effective threshold $T$ follows such that for $t \geq T$, $f(t)$ is
%positive.
\end{proof}

\begin{lemma}\label{lem: oneOscTwo}
Let $A, B, a, b, c, r$ be real algebraic numbers such that 
$a, b, c, r > 0$,  $A, B \neq 0$.
Let also $\varphi_1,\varphi_2,\varphi_3\in\re$ be
such that
$e^{i\varphi_1},e^{i\varphi_2},e^{i\varphi_3}\in\lalg$.  
Define the exponential polynomial $f$ by
\[
f(t) = 1-\cos(ct + \varphi_3) + e^{-rt}(A\cos(at+\varphi_1) + B\cos(bt+\varphi_2)).
\]
Then it is decidable whether $f$ has infinitely many
zeros. 
%Moreover, if $f(t)$ has only finitely many zeros, then there
%exists an effective threshold $T$ such that all zeros of $f(t)$ are
%contained in $[0,T]$.
\end{lemma}
\begin{proof}
%As in the first paragraph of Lemma \ref{lem: layered}, 
%we consider first the magnitudes of $C$ and $D$. If $|C|>|D|$,
%then $f(t)$ oscillates around $0$ and has infinitely many zeros.
%If $|C|<|D|$, then for some effective $T$, $f(t)$ has the same sign
%as $D$ on $[T,\infty)$. Finally, if $|D|=|C|$, then divide through
%by $D$, scaling constants and possibly replacing $\varphi_3$ by 
%$\varphi_3+\pi$, and assume without loss of generality that the given
%function is of the form:
%\[ 
%f(t) = 1-\cos(ct + \varphi_3) + e^{-rt}(A\cos(at+\varphi_1) + B\cos(bt+\varphi_2)).
%\]
We argue the function is infinitely often positive and infinitely often negative
by looking at the values of $t$ for which the dominant term $1-\cos(ct+\varphi_3)$
vanishes. This happens precisely at the times $t = -(\varphi_3+2k\pi)/c$ for $k\in\zed$,
giving rise to a discrete restriction of $f$:
\begin{align*}
g(k) \defn e^{r\varphi_3}\left(e^{2\pi r}\right)^k
\left(
A\cos\left( k\frac{2\pi a}{c} - \frac{a\varphi_3}{c}+\varphi_1 \right) +\right.\\
\left.
B\cos\left( k\frac{2\pi b}{c} - \frac{b\varphi_3}{c}+\varphi_2 \right)
\right).
\end{align*}
This is a linear recurrence sequence over $\re$ of order $4$, 
with characteristic roots $e^{2\pi(r\pm ia/c)}$ and $e^{2\pi(r\pm ib/c)}$. 
In particular, it has no real dominant characteristic root. It is well-known 
that real-valued  linear recurrence sequences with no dominant real 
characteristic root are infinitely often positive and infinitely often negative:
see for example \cite[Theorem 7.1.1]{gyori}. Therefore, by continuity, 
$f$ must have infinitely many zeros.
\end{proof}

\begin{lemma}\label{lem: oneOscOneRep}
Let $A, B, a, b, r$ be real algebraic numbers such that 
$a, b, r > 0$,  $A \neq 0$.
Let also $\varphi_1,\varphi_2,\varphi_3\in\re$ be
such that
$e^{i\varphi_1},e^{i\varphi_2},e^{i\varphi_3}\in\lalg$.  
Define the exponential polynomial $f$ by
\[
f(t) = 1-\cos(at + \varphi_1)
+ e^{-rt}(At\cos(bt + \varphi_2) + B\cos(bt+\varphi_3)).
\]
Then it is decidable whether $f$ has infinitely many
zeros. 
%Moreover, if $f(t)$ has only finitely many zeros, then there
%exists an effective threshold $T$ such that all zeros of $f(t)$ are
%contained in $[0,T]$.
\end{lemma}
\begin{proof}
If $a/b\in\rats$, then the claim follows immediately from Theorem
\ref{thm: mtargument}. If $a/b\not\in\rats$, then by Lemma \ref{lem: density},
it will happen infinitely often that $1-\cos(at+\varphi_1) = 0$ and
$At\cos(bt+\varphi_2) < -|A|t/2$. Then clearly $f(t) < 0$ infinitely often.
Since $f(t)>0$ infinitely often as well, due to the non-negative dominant
term $1-\cos(at+\varphi_1)$, it follows that $f$ has infinitely many zeros.
\end{proof}

\subsection{Two dominant oscillations}\label{sec:twoOsc}

\begin{lemma}\label{lem: twoOsc}
Let $A, B, C, a, b, r$ be real algebraic numbers such that 
$a, b, r > 0$, $a\neq b$ and  $A, B, C \neq 0$.
Let also $\varphi_1,\varphi_2\in\re$ be
such that
$e^{i\varphi_1},e^{i\varphi_2}\in\lalg$.  
Define the exponential polynomial $f$ by
\[
f(t) = A\cos(at + \varphi_1) + B\cos(bt+\varphi_2) + C
+ e^{-rt}F(t).
\]
where $F$ is an exponential polynomial whose dominant
characteristic roots are purely imaginary.
Suppose also $f$ has order at most $8$.
It is decidable whether $f$ has infinitely many
zeros. 
%and moreover, if $f(t)$ has only finitely many zeros, then there
%exists an effective threshold $T$ such that all zeros of $f(t)$ are
%contained in $[0,T]$.
\end{lemma}
\begin{proof}
If the frequencies $a, b$ of the dominant term's oscillations are linearly
independent over $\rats$, then the claim follows immediately
by Theorem \ref{blondelubcs}. Therefore, assume $na - mb = 0$ for some
$n,m\in\mathbb{N}^{+}$. Notice that $a\neq b$ guarantees $n\neq m$. 
We perform the change of variable $t \rightarrow tm/a$, so that:
\[ 
f(t) =  A\cos(mt + \varphi_1) + B\cos(nt + \varphi_2) + C + e^{-rmt/a}F(tm/a).
\]
Using the standard trigonometric identities, we express
the dominant term as a polynomial in $\sin(t), \cos(t)$:
\[
f(t) = P(\sin(t),\cos(t)) + e^{-rmt/a}F(tm/a),
\]
where $P\in(\ra)[x,y]$ has effectively computable 
coefficients. It is clear that the dominant term is periodic. 
It is immediate from the definition of exponential polynomials
and the premise of the Lemma that $F(tm/a)\defn F_2(t)$ is an 
exponential polynomial in $t$, of the same order as $F(t)$, 
also with purely imaginary dominant characteristic roots.
Let $\alpha(t) \defn P(\sin(t),\cos(t))$, $r_2 \defn rm/a>0$
and $\beta(t) \defn e^{-rmt/a}F(tm/a) = e^{-r_2t}F_2(t)$.

We are now interested in the extrema of $P(\sin(t),\cos(t))$. Let
\begin{align*}
M_1 \defn \min_{x^2 + y^2 = 1} P(x, y) = \min_{t\geq 0} \alpha(t), \\
M_2 \defn \max_{x^2 + y^2 = 1} P(x, y) = \max_{t\geq 0} \alpha(t).
\end{align*}
We can construct defining formulas $\phi_1(u), \phi_2(u)$
in the first-order language $\lang$ of real closed fields for 
$M_1, M_2$, so that each $\phi_j(u)$ holds
precisely for the valuation $u = M_j$.
Then performing quantifier
elimination on these formulas using Renegar's algorithm \cite{Renegar},
we convert $\phi_1,\phi_2$ into the form
\[ \phi_j(u) \equiv \bigvee_l\bigwedge_k P_{l,k} (u) \sim_{l,k} 0, \]
where $P_{l,k}$ are polynomials with integer coefficients and each
$\sim_{l,k}$ is either $<$ or $=$. Now $\phi_j(u)$ must have a satisfiable
disjunct. Using the decidability of the theory $\thr$, we can
readily identify this disjunct.
Moreover, since $\phi_j(u)$ has a unique satisfying valuation,
namely $u=M_j$, this disjunct must contain at least one equality predicate.
It follows immediately that $M_1, M_2$ are algebraic. Moreover,
we can effectively compute from $\phi_j(u)$ a representation for $M_j$ 
consisting of its minimal polynomial and a sufficiently accurate 
rational approximation to distinguish $M_j$ from its Galois conjugates.
By an analogous argument, the pairs
$(\sin(t), \cos(t))$ at which $P(\sin(t), \cos(t))$ achieves the extrema
$M_1, M_2$ are also algebraic and effectively computable.

We now perform a case analysis on the signs of $M_1$ and $M_2$.
\begin{itemize}
\item First, if $0 < M_1 \leq M_2$, then $f(t)$ cannot have infinitely many zeros:
if $t$ is large enough to ensure $|\beta(t)| < M_1$,
we have $f(t)>0$. 
\item Second, if $M_1 \leq M_2 < 0$, then by the same reasoning,
the function will ultimately be strictly negative. 
\item Third, if $M_1 < 0 < M_2$,
then $f$ oscillates around $0$: for all $t$ such that 
$\alpha(t) = M_1 < 0$ and large enough
to ensure $|\beta(t)| < |M_1|$, we will have $f(t) < 0$, and 
similarly, for large enough $t$ such that $\alpha(t)=M_2>0$,
we will have $f(t)>0$, so the function must have infinitely
many zeros.
\item Next, we argue that the case $M_1=M_2=0$
is impossible. Indeed, if $M_1=M_2=0$, then 
$\alpha(t)=P(\sin(t),\cos(t))$ is identically zero,
and the same holds for all derivatives of $\alpha(t)$.
Thus, from $\alpha'(t) \equiv \alpha'''(t) \equiv 0$,
we have 
\begin{align*} 
0 & \equiv -Am\sin(mt+\varphi_1)-Bn\sin(nt+\varphi_2), \\
0 & \equiv Am^3\sin(mt+\varphi_1)+Bn^3\sin(nt+\varphi_2). \\
\end{align*}
Multiplying the first identity through by $m^2$ and summing,
we have
\[ Bn\sin(nt+\varphi_2)(n^2-m^2) \equiv 0. \]
By the premise of the Lemma, $B\neq 0$, so $n(n-m)(n+m)=0$,
which is a contradiction.
\item Finally, only the symmetric cases $M_1 < M_2 = 0$ 
and $0 = M_1 < M_2$ remain. Without loss of 
generality, by replacing $f$ by $-f$ if 
necessary, we need only consider the case
$0=M_1 < M_2$.
\end{itemize}

Thus, assume $0=M_1 < M_2$. We now move our attention to
the possible forms of $F_2$. Since $f$ has order at most
$8$, it follows that $F_2$
has order at most $3$. Thus, there are three possibilities
for the set of dominant characteristic roots of $F_2$: $\{0\}$,
$\{\pm ic\}$, or $\{0,\pm ic\}$, for some positive $c\in\ra$.
We consider each of these cases in turn.

First, if $F_2$ only has the real dominant eigenvalue $0$,
then $F_2$ is ultimately positive or ultimately negative,
depending on the sign of the most significant term of $F_2$.
%with an effectively computable threshold.
Ultimate positivity of $F_2$ entails ultimate positivity
of $f$ as well, since $P(\sin(t),\cos(t))\geq 0$ everywhere,
whereas an ultimately negative $F_2$ makes $f$ change sign
infinitely often.

Second, assume the dominant characteristic roots of $F_2$ 
are $\{\pm ic\}$, so that
\[ f(t) = P(\sin(t),\cos(t)) + e^{-r_2t}\left(D\cos(ct+\varphi_3) 
+ Ee^{-r_3t}\right) \]
for some $r_3>0$ and $\varphi_3\in\re$ such that $e^{i\varphi_3}\in\alg$.
Without loss of generality, we can assume $c\not\in\rats$,
since otherwise, we are done by Theorem \ref{thm: mtargument}. 
But by Lemma \ref{lem: density},
it will happen infinitely often that $P(\sin(t),\cos(t)) = 0$
and $D\cos(ct+\varphi_3) < -|D|/2$, say. For large enough such
$t$, $|Ee^{-(r_2+r_3)t}| < |D|/4$, so we conclude that $f$
is infinitely often negative, and hence has infinitely many zeros.

Third, assume the dominant characteristic roots of $F_2(t)$ are $\{0, \pm ic \}$,
so that
\[ f(t) = P(\sin(t),\cos(t)) + e^{-r_2t}(D\cos(ct+\varphi_3) + E). \]
We again assume $c\not\in\rats$, since otherwise the claim follows
from Theorem \ref{thm: mtargument}. Let $M_3 \defn E - |D| 
= \min_{t\geq 0} F_2(t)$. If
$M_3 > 0$, then $f(t)$ clearly has no zeros. If $M_3 < 0$, then
there exists a non-trivial interval $I\subseteq [0,2\pi)$ such that
if $ct + \varphi_3 \bmod 2\pi \in I$, then $F_2(t) < 0$. Since $c\not\in
\rats$, Lemma \ref{lem: density} guarantees that $F_2(t) < 0 = P(\sin(t),\cos(t))$
happens infinitely often,
so $f$ must have infinitely many zeros. Finally, if 
$M_3 = 0$, we argue that $f$ is ultimately positive.
Indeed, since $P(\sin(t),\cos(t))$ and $F_2(t)$
are both non-negative everywhere, $f(t) = 0$ can only happen
if $P(\sin(t),\cos(t))=D\cos(ct+\varphi_3)+E=0$. This,
however, would entail $e^{it}\in\alg$ and $e^{ict}\in\alg$,
which contradicts the Gelfond-Schneider Theorem, since $c\not\in\rats$.
Thus, we conclude $f$ has no zeros.
\end{proof}

\subsection{Three dominant oscillations}\label{sec:threeOsc}

\begin{lemma}\label{lem: threeOsc}
Let $A, B, C, a, b, c$ be real algebraic numbers such that 
$a, b, c > 0$ and  $A, B, C \neq 0$.
Let also $\varphi_1,\varphi_2,\varphi_3\in\re$ be
such that
$e^{i\varphi_1},e^{i\varphi_2},e^{i\varphi_3}\in\lalg$.  
Define the exponential polynomial $f$ by
\[
f(t) = A\cos(at + \varphi_1) + B\cos(bt+\varphi_2) +
C\cos(ct+\varphi_3) + D.
\]
It is decidable whether $f$ has infinitely many
zeros.
% and moreover, if $f(t)$ has only finitely many zeros, then there
%exists an effective threshold $T$ such that all zeros of $f(t)$ are
%contained in $[0,T]$.
\end{lemma}
\begin{proof}
The argument consists of three cases, depending on the linear dependencies
over $\rats$ satisfied by $a,b$ and $c$.

\emph{Case I.} First, if $a,b,c$ are linearly independent over $\rats$,
then the claim follows directly from Theorem \ref{blondelubcs}.

\emph{Case II.} Second, suppose that $a,b,c$ are all rational multiples
of one another:
\[
b = \frac{n}{m}a, \, c = \frac{k}{l}a\, \mbox{ where $n,m,k,l\in\nat^{+}$}.
\]
We make the change of variable $t\rightarrow tml$ to obtain:
\begin{align*}
f(t)  = \; & A\cos((at)ml + \varphi_1) + B\cos((at)nl+\varphi_2) \\
& + C\cos((at)km+\varphi_3) + D \\
 = \; & P(\sin(at),\cos(at)),
\end{align*}
where $P\in\alg[x,y]$ is a polynomial obtained using the standard
trigonometric identities. It is now clear that $f$ is periodic,
so it has either no zeros or infinitely many zeros. Let 
\begin{align*}
M_1 \defn \min_{x^2 + y^2 = 1} P(x, y) = \min_{t\geq 0} f(t), \\
M_2 \defn \max_{x^2 + y^2 = 1} P(x, y) = \max_{t\geq 0} f(t).
\end{align*}
Using the same reasoning as in Lemma \ref{lem: twoOsc},
we see that $M_1, M_2$ are algebraic and effectively computable:
simply construct defining formulas in the first-order language
$\lang$ of real closed fields, and then perform quantifier elimination
using Renegar's algorithm \cite{Renegar}. Then $f$ clearly
has infinitely many zeros if and only if $M_1\leq 0\leq M_2$.

\emph{Case III.} Finally, suppose that $a,b,c$ span a $\rats$-vector
space of dimension $2$, so that $a,b,c$ satisfy a single linear dependence
$am + bn + cp = 0$ where $m, n, p\in\zed$ are coprime. At most one of the ratios
$a/b$, $a/c$ and $b/c$ is rational (otherwise we have 
$\dim\spa\{a,b,c\}=1$), so assume without loss of generality that
$a/c\not\in\rats$ and $b/c\not\in\rats$.

Define the set
\begin{align*} 
\mathbb{T} & \defn
\left\{ x \in [0,2\pi)^3
\,\middle|\,
\forall u\in\zed^3\,.\,
u\cdot (a,b,c) \in 2\pi\zed \Rightarrow
u\cdot x \in 2\pi\zed
\right\} \\
& = \left\{ (x_1,x_2,x_3) \in [0,2\pi)^3
\,\middle|\, 
mx_1 + nx_2 + px_3  \in 2\pi\zed \right\}
\end{align*}
Notice that if $mx_1 + nx_2 + px_3 = 2k\pi$ for some $x_1,x_2,x_3$, then $k\leq
|m| + |n| + |p|$, so $\mathbb{T}$ partitions naturally into finitely many
subsets:  $\mathbb{T}=\bigcup_{k=1}^N \mathbb{T}_k$, where
\[ 
\mathbb{T}_k \defn 
\left\{ (x_1,x_2,x_3) \in [0,2\pi)^3
\,\middle|\, 
mx_1 + nx_2 + px_3 = 2k\pi \right\}.
\]
Consider the trajectory $h(t)\defn\left\{ (at,bt,
ct)\bmod 2\pi\,\middle|\,t\geq 0\right\}$.
Define also the sets $R \defn \{ h(2k\pi)\,|\,k\in\nat \}$ and 
$H\defn \{ h(t)\,|\,t\geq 0 \}$. Because of the linear dependence satisfied 
by $a,b,c$, it is easy to see that $R\subseteq H\subseteq\mathbb{T}$.
By Kronecker's Theorem, $R$ is a dense subset of $\mathbb{T}$, so clearly $H$
must be a dense subset of $\mathbb{T}$ as well.

Now define the function
\begin{align*}
F(x_1,x_2,x_3) \defn \; & A\cos(x_1+\varphi_1) + B\cos(x_2+\varphi_2) \\
& + C\cos(x_3+\varphi_3) + D,
\end{align*}
so that the image of $f$ is exactly $\{F(x_1,x_2,x_3)\,|\,(x_1,x_2,x_3)\in H\}$.
Let also the extrema of $F$ over $\mathbb{T}$ be:
\begin{align*} 
M_1 & \defn \min_{\mathbb{T}} F(x_1,x_2,x_3), \\
M_2 & \defn \max_{\mathbb{T}} F(x_1,x_2,x_3). \\
\end{align*}
Both of these values are algebraic and can be computed using quantifier elimination 
in the first-order language $\lang$ of the real numbers: just use separate variables
for $\cos(x_j),\sin(x_j)$ and apply the standard trigonometric identities to convert
the linear dependence on $x_1,x_2,x_3$ into a polynomial dependence between $\cos(x_j),\sin(x_j)$.

Now, by the density of $H$ in $\mathbb{T}$,
if $M_1 < 0 < M_2$, then $f$ must clearly be infinitely often positive 
and infinitely often negative, so it must have infinitely many zeros. 
The case $M_1 < 0 = M_2$ is symmetric to $0 = M_1 < M_2$ (just replace
$f$ and $F$ by $-f$ and $-F$, respectively), so without loss generality, we can assume
$0 = M_1 < M_2$. In this case, we argue that $f$ has no zeros, that is,
even though $F$ vanishes on some points in $\mathbb{T}$, none of these points
appear in the dense subset $H$. Indeed, consider the set
\begin{align*} 
Z \defn & \left\{ (\cos(x_1),\sin(x_1),\dots,\cos(x_3),\sin(x_3))\,\middle|\, \right.\\
& \left.(x_1,x_2,x_3)\in\mathbb{T}, F(x_1,x_2,x_3)=0\right\}.
\end{align*}
Note that $Z$ is clearly semi-algebraic, as one can directly write a 
defining formula in $\lang$ from $F(x_1,x_2,x_3)=0$ and $mx_1 + nx_2 + px_3 \in 2\pi\zed$.
Moreover, by the Zero-Dimensionality Lemma \cite[Lemma 10]{ouaknine2014positivity},
the function $F(x_1,x_2,x_3)$ achieves its minimum $M_1=0$ at only finitely 
many points in $\mathbb{T}_k$, for each $k$. Since $\mathbb{T}$ 
is the union of finitely many $\mathbb{T}_k$, we immediately have that $Z$ is
finite. By the Tarski-Seidenberg Theorem, projecting $Z$ to any fixed component
will also give a finite, semi-algebraic subset of $\re$, that is, a finite subset
of $\alg$. Thus, we have shown that if $F(x_1,x_2,x_3)=0$, then $e^{ix_j}\in\alg$
for all $j=1,2,3$. Now if $f(t) = 0$ for some $t\geq 0$, then we must have
$e^{ati},e^{cti}\in\alg$, which by the Gelfond-Schneider Theorem entails
$a/c\in\rats$, a contradiction.
\end{proof}

\subsection{One repeated oscillation}\label{sec:repOsc}

\begin{lemma}\label{lem: repOsc}
Let $A, B, C, D, a, r$ be real algebraic numbers such that 
$a, r > 0$ and  $A \neq 0$.
Let also $\varphi_1,\varphi_2\in\re$ be
such that
$e^{i\varphi_1},e^{i\varphi_2}\in\lalg$.  
Define the exponential polynomial $f$ by
\[
f(t) = t(A\cos(at+\varphi_1) + B) + (C\cos(at + \varphi_2) + D) + e^{-rt}F(t)
\]
where $F$ is an exponential polynomial with 
purely imaginary dominant characteristic roots.
Suppose also that $f$ has order at most $8$.
It is decidable whether $f$ has infinitely many
zeros. 
%and moreover, if $f(t)$ has only finitely many zeros, then there
%exists an effective threshold $T$ such that all zeros of $f(t)$ are
%contained in $[0,T]$.
\end{lemma}
\begin{proof}
Since $f$ has order no greater than $8$,
it follows that $F$ has order at most $2$. Therefore,
$F(t)$ must be of the form $E\cos(bt+\varphi_3)$
for some $E,b\in\ra$, $b>0$, such that $a/b\not\in\rats$, and some
$\varphi_3$ such that $e^{i\varphi_3}\in\alg$, since otherwise
the imaginary parts of the characteristic roots of $f$ are pairwise linearly
dependent over $\rats$, so our claim is proven immediately by 
Theorem \ref{thm: mtargument}.

Consider first the magnitudes of $A$ and $B$.
If $|A|>|B|$, then the term $tA\cos(at+\varphi_1)$ makes $f$ 
change sign infinitely often, whereas if $|B|>|A|$, then
for $t$ large enough, the term $tB$ makes $f$ ultimately
positive or ultimately negative, depending on the sign of $B$.
Thus, we can assume $|A|=|B|$. Dividing $f$ by $B$, and
replacing $\varphi_1$ by $\varphi_1+\pi$ if necessary, we can assume
the function has the form:
\begin{align*}
f(t) = \; & t(1-\cos(at+\varphi_1)) \\
& + (C\cos(at+\varphi_2)+D) + e^{-rt}E\cos(bt+\varphi_3).
\end{align*}
Considering the dominant term, it is clear that 
$f$ is infinitely often positive. Let 
$\alpha(t)\defn t(1-\cos(at+\varphi_1))$, 
$\beta(t)\defn C\cos(at+\varphi_2)+D$ and
$\gamma(t)\defn e^{-rt}E\cos(bt+\varphi_3)$.

We now focus on the sign of the term $\beta(t)$ at the positive
\emph{critical times} $t_j\defn -\varphi_1/a + 2j\pi/a$ ($j\in\zed$)
when $1-\cos(at+\varphi_1)$ vanishes.
Notice that $\beta(t_j)=C\cos(\varphi_2-\varphi_1)+D\defn M$ is
independent of $j$.
First, if $M < 0$, then for all $t_j$
large enough, $f(t_j)<0$, so the function must have infinitely many
zeros. Second, if $M=0$, then by the linear independence of $a,b$ and 
Lemma \ref{lem: density}, we have 
$\alpha(t_j) = \beta(t_j) = 0 > \gamma(t_j)$ for infinitely many $t_j$, so we
can conclude $f$ has infinitely many zeros.

Finally, suppose $M > 0$. We will prove that $f$ is ultimately positive.
For each $t_j$, define the \emph{critical region} 
$[t_j-\delta_j, t_j+\delta_j]$, given by 
\[
\delta_j \defn \frac{2\sqrt{|C|+|D|}}{a\sqrt{t_{j-1}}}.
\]
From here onwards, we only consider $t$ large enough for any two
adjacent critical regions to be disjoint. The argument consists of two parts:
first we show $f(t) > 0$ for all large enough $t$ outside all critical
regions, and then we show $f(t) > 0$ for large enough $t$ in a critical region.

Suppose $t$ is outside all critical regions and let
$j$ minimise $|t-t_j|$. Since the distance between critical points is
$2\pi/a$ by construction, we have $a|t-t_j|\leq \pi$. Therefore,
\[
\frac{|a(t-t_j)|^2}{2} \leq 1-\cos(at - at_j) = 1-\cos(at + \varphi_1).
\]
On the other hand, we have the following chain of inequalities:
\begin{align*}
& \frac{|a(t-t_j)|^2}{2} \\
> & \mbox{ $\{$ $|t-t_j|>\delta_j$ $\}$ }  \\
& \frac{(a\delta_j)^2}{2} \\
= & \mbox{ $\{$ definition of $\delta_j$ $\}$ } \\
& \frac{2(|C|+|D|)}{t_{j-1}}  \\
> & \mbox{ $\{$ by $t>t_{j-1}$ $\}$ } \\
& \frac{2(|C|+|D|)}{t} \\
\geq & \mbox{ $\{$ triangle inequality and $|\cos(x)|\leq1$ $\}$ } \\ 
& \frac{|C|+|D|}{t} + \frac{|C\cos(at+\varphi_2)+D|}{t}. \\
\end{align*}
Combining, we have 
\begin{align*}
\alpha(t)+\beta(t) & \geq \alpha(t)-|\beta(t)| \\
& = t(1-\cos(at+\varphi_1)) - |C\cos(at + \varphi_2)+D| \\
& \geq |C|+|D|.
\end{align*}
Thus, if $t$ is large enough to ensure $|\gamma(t)| < |C|+|D|$, we have $f(t)>0$ outside critical regions.

For the second part of the argument, we consider $t$ in 
critical regions. Notice that the values of $\beta(t)$ on 
$[t_j-\delta_j,t_j+\delta_j]$ are independent of the choice 
of $t_j$. Moreover, we have $\beta(t_j) = M > 0$, so there 
exists some $\epsilon > 0$ such that for all $t\in[t_j-\epsilon, 
t_j+\epsilon]$, we have $\beta(t) \geq M/2$, say. Now for any 
critical point $t_j$ chosen large enough, we will have 
$[t_j - \delta_j,t_j + \delta_j]\subseteq [t_j -\epsilon, 
t_j+\epsilon]$, so $\beta(t) > M/2$ on the entire critical 
region. Let also $t_j$ be large enough so that for any $t$ in 
the critical region, we have $|\gamma(t)| < M/2$. Then we have 
$f(t) = \alpha(t) + \beta(t) + \gamma(t) \geq \beta(t) - |\gamma(t)| 
> 0$, completing the claim.
\end{proof}
\section{Hardness at Order 9}\label{sec:hardness}

Diophantine approximation is a branch of number theory concerned with
approximating real numbers by rationals.  A central role is played in
this theory by the notion of \emph{continued fraction expansion},
which allows to compute a sequence of rational approximations to a
given real number that is optimal in a certain well-defined sense.
For our purposes it suffices to note that the behaviour of the
continued fraction expansion of a real number $a$ is closely related
to the following two constants associated with $a$.  The
\emph{Lagrange constant} (or \emph{homogeneous Diophantine
  approximation constant}) of $a$ is defined by
\[
L_{\infty}(a) = \inf\left\{ c : \left| a - \frac{n}{m} \right| < \frac{c}{m^2}
\mbox{ for infinitely many $m,n\in\mathbb{Z}$}\right\}.
\]
By definition $L_\infty(a)$ is a non-negative real number.

% Following the terminology of Lagarias and Shallit~\cite{LS97},  
% the \emph{(homogeneous Diophantine approximation) type} of $a$ is defined by
% \[
% L(a) = \inf\left\{ c : \left| a - \frac{n}{m} \right| < \frac{c}{m^2}
% \mbox{ for some $m,n\in\mathbb{Z}$}\right\}.
% \]
A real number $a$ is called \emph{badly approximable} if
$L_\infty(a)>0$.  The badly approximable
numbers are precisely those whose continued fraction expansions have bounded
partial quotients.

Khinchin showed in 1926 that almost all real numbers (in the
measure-theoretic sense) have Lagrange constant equal to
zero.  However, information on the Lagrange constants of
specific numbers or classes of numbers has proven to be elusive.  In
particular, concerning algebraic numbers, Guy~\cite{Guy04} asks
% \begin{quote} \emph{[\ldots] no continued fraction development of
% an algebraic number of higher degree than the second is known.
% It is not even known if such a development has bounded elements.}
% \end{quote}
\begin{quote}
  Is there an algebraic number of degree greater than two whose simple
  continued fraction expansion has unbounded partial quotients?  Does
  every such number have unbounded partial quotients?
\end{quote}
The above question can equivalently be formulated in terms of whether
any algebraic number of degree greater than two has stricly positive
Lagrange constant or whether all such numbers have Lagrange constant
$0$.

Recall that a real number $a$ is \emph{computable} if there is an
algorithm which, given any rational $\varepsilon>0$ as input, returns
a rational $q$ such that $|q-x|<\varepsilon$.  We can now state the
main result of the section.

In this section, we will show that a decision procedure
for the Infinite Zeros Problem would yield the computability of
$L_{\infty}(a)$ for all $a\in\ra$.

Fix positive $a\in\ra$, $c\in\rats$ and define the functions:
\begin{align*}
f_1(t) & \defn e^t(1-\cos(t)) + t (1 - \cos(at)) - c\sin(at),  \\
f_2(t) & \defn e^t(1-\cos(t)) + t (1 - \cos(at)) + c\sin(at),  \\
f(t) & \defn e^t(1-\cos(t)) + t (1 - \cos(at)) - c|\sin(at)| \\
& = \min\{f_1(t),f_2(t)\}.
\end{align*}
It is easy to see that $f_1$ and $f_2$ are exponential
polynomials of order 9, with six characteristic roots: three simple ($1$ and
$1\pm i$) and three repeated ($0$ and $\pm ai$). Thus, the
problem of determining whether $f_j$ has infinitely
many zeros is an instance of the Infinite Zeros Problem.
Moreover, it is easy to check that 
$f$ has infinitely many zeros if and only
if at least one of $f_1$ and $f_2$ has infinitely many
zeros.

We will first state two lemmas which show a connection between
the existence of infinitely many zeros of $f$ and the Lagrange 
constant of $a$. We defer the proofs to Appendix \ref{sec:hardnessProofs}. 
\setcounter{venchoRestateLemma1}{\arabic{theorem}}
\begin{lemma}\label{hardForward}
Fix $a\in\ra$ and $\varepsilon,c\in\rats$ with $a,c>0$
and $\varepsilon\in(0, 1)$. If $f(t) = 0$ for infinitely many
$t\geq 0$, then $L_{\infty}(a)\leq c/2\pi^2(1-\varepsilon)$.
%There exists an effective threshold $T$,
%dependent on $a, c, \varepsilon$, such that if $f(t)=0$ for some
%$t\geq T$, then $L(a)\leq c/2\pi^2(1-\varepsilon)$.
\end{lemma}
\setcounter{venchoRestateLemma2}{\arabic{theorem}}
\begin{lemma}\label{hardBackward}
Fix $a\in \ra$ and $\varepsilon,c\in\rats$ with $a,c>0$ and
$\varepsilon\in(0, 1)$. 
If $L_{\infty}(a)\leq c(1-\varepsilon)/2\pi^2$, then $f(t)=0$
for infinitely many $t\geq 0$.
\end{lemma}

We now use the above lemmas to derive an algorithm to compute $L_{\infty}(a)$
using an oracle for the Infinite Zeros Problem, establishing our central
hardness result:
\begin{theorem}\label{thm:hardnessMain}
Fix a positive real algebraic number $a$. If the Infinite Zeros
Problem is decidable for instances of order $9$,
then $L_{\infty}(a)$ may be computed to within arbitrary precision.
\end{theorem}
\begin{proof}
Suppose we know $L_{\infty}(a)\in[p, q]$ for non-negative
$p,q\in\rats$. Choose $c\in\rats$ with $c>0$ and $\varepsilon\in\rats$
with $\varepsilon\in(0,1)$ such that
\[
p < \frac{c(1-\varepsilon)}{2\pi^2} < \frac{c}{2\pi^2(1-\varepsilon)}
< q.
\]
Write $A \defn c(1-\varepsilon)/2\pi^2$ and $B \defn c/2\pi^2(1-\varepsilon)$.
%Calculate the maximum of the thresholds $T$ required by Lemmas
%\ref{hardForward} and \ref{hardBackward}.  Check for all denominators
%$m\leq T/2\pi$ whether there exists a numerator $n$ such that $n, m$
%witness $L(a) \leq A$.  If so, then continue the approximation procedure
%recursively with confidence interval $[p, A]$.  
%Otherwise, 
Use the oracle for the Infinite Zeros Problem to determine whether at least
one of $f_1, f_2$ has infinitely many zeros.
%a zero on $[T, \infty)$. 
If this is the
case, then $f$ also has 
%a zero on $[T, \infty)$, 
infinitely many zeros, 
so by Lemma
\ref{hardForward}, $L_{\infty}(a)\leq B$ and we continue the
approximation recursively on the interval $[p, B]$. If not,
then $L(a)\geq A$ by Lemma \ref{hardBackward}, so we continue on
the interval $[A, q]$. Notice that in this procedure, one can
choose $c, \varepsilon$ at each stage in such a way that the confidence
interval shrinks by at least a fixed factor,
whatever the outcome of the oracle invocations. It follows
therefore that $L_{\infty}(a)$ can be approximated to within arbitrary
precision.
\end{proof}

\subparagraph*{Acknowledgements} Ventsislav Chonev is supported by
Austrian Science Fund (FWF) NFN Grant No S11407-N23 (RiSE/SHiNE), ERC
Start grant (279307: Graph Games), and ERC Advanced Grant (267989:
QUAREM).  Jo\"{e}l Ouaknine is supported by ERC grant AVS-ISS
(648701).  James Worrell is supported by EPSRC grant
EP/N008197/1.

%\bibliographystyle{abbrvnat}
%\bibliography{thesis}

\begin{thebibliography}{19}
\providecommand{\natexlab}[1]{#1}
\providecommand{\url}[1]{\texttt{#1}}
\expandafter\ifx\csname urlstyle\endcsname\relax
  \providecommand{\doi}[1]{doi: #1}\else
  \providecommand{\doi}{doi: \begingroup \urlstyle{rm}\Url}\fi

\bibitem[Baker(1975)]{Baker75}
A.~Baker.
\newblock \emph{Transcendental number theory.}
\newblock Cambridge University Press, Cambridge, 1975.

\bibitem[Bell et~al.(2010)Bell, Delvenne, Jungers, and Blondel]{BDJB10}
P.~C. Bell, J.-C. Delvenne, R.~M. Jungers, and V.~D. Blondel.
\newblock The {C}ontinuous {S}kolem-{P}isot {P}roblem.
\newblock \emph{Theoretical Computer Science}, 411\penalty0 (40-42):\penalty0
  3625--3634, 2010.
\newblock ISSN 0304-3975.

\bibitem[Berstel and Mignotte(1976)]{BerstelMignotte76}
J.~Berstel and M.~Mignotte.
\newblock Deux propri{\'e}t{\'e}s d{\'e}cidables des suites r{\'e}currentes
  lin{\'e}aires.
\newblock \emph{Bulletin de la Soci{\'e}t{\'e} Math{\'e}matique de France},
  104:\penalty0 175--184, 1976.

\bibitem[Cohen(1993)]{Coh93}
H.~Cohen.
\newblock \emph{A Course in Computational Algebraic Number Theory}.
\newblock Springer, 1993.

\bibitem[Gelfond(1934)]{gelfonda}
A.~O. Gelfond.
\newblock On {H}ilbert's seventh problem.
\newblock In \emph{Dokl. Akad. Nauk. SSSR}, volume~2, pages 1--6, 1934.

\bibitem[Gelfond and Vinogradov(1934)]{gelfondb}
A.~O. Gelfond and I.~Vinogradov.
\newblock Sur le septieme probleme de {H}ilbert.
\newblock \emph{Bull. Acad. Sci. URSS}, pages 623--634, 1934.

\bibitem[Guy(2004)]{Guy04}
R.~Guy.
\newblock \emph{Unsolved Problems in Number Theory}.
\newblock Springer, third edition, 2004.

\bibitem[Gy{\H{o}}ri and Ladas(1991)]{gyori}
I.~Gy{\H{o}}ri and G.~Ladas.
\newblock \emph{Oscillation Theory of Delay Differential Equations: with
  Applications}.
\newblock Oxford mathematical monographs. Oxford University Press, 1991.
\newblock ISBN 9780198535829.
\newblock URL \url{https://books.google.co.uk/books?id=6CnvAAAAMAAJ}.

\bibitem[Hardy and Wright(1999)]{hardy}
G.~Hardy and E.~Wright.
\newblock An introduction to the theory of numbers.
\newblock \emph{Oxford}, 1:\penalty0 979, 1999.

\bibitem[Macintyre and Wilkie(1996)]{macintyreWilkie}
A.~Macintyre and A.~J. Wilkie.
\newblock On the decidability of the real exponential field.
\newblock 1996.

\bibitem[Marker(2002)]{marker}
D.~Marker.
\newblock \emph{Model Theory: An Introduction}.
\newblock Graduate Texts in Mathematics. Springer, 2002.

\bibitem[Ouaknine and Worrell(2014)]{ouaknine2014positivity}
J.~Ouaknine and J.~Worrell.
\newblock On the positivity problem for simple linear recurrence sequences.
\newblock In \emph{Automata, Languages, and Programming}, pages 318--329.
  Springer, 2014.

\bibitem[Pan(1996)]{panApproximatingRoots}
V.~Pan.
\newblock Optimal and nearly optimal algorithms for approximating polynomial
  zeros.
\newblock \emph{Computers {\&} Mathematics with Applications}, 31\penalty0
  (12):\penalty0 97 -- 138, 1996.

\bibitem[Renegar(1992)]{Renegar}
J.~Renegar.
\newblock On the computational complexity and geometry of the first-order
  theory of the reals. part i: Introduction. preliminaries. the geometry of
  semi-algebraic sets. the decision problem for the existential theory of the
  reals.
\newblock \emph{Journal of Symbolic Computation}, 13\penalty0 (3):\penalty0 255
  -- 299, 1992.
\newblock ISSN 0747-7171.
\newblock \doi{http://dx.doi.org/10.1016/S0747-7171(10)80003-3}.

\bibitem[Schneider(1935{\natexlab{a}})]{schneidera}
T.~Schneider.
\newblock Transzendenzuntersuchungen periodischer {F}unktionen {I}.
  {T}ranszendenz von {P}otenzen.
\newblock \emph{Journal f{\"u}r die reine und angewandte Mathematik},
  172:\penalty0 65--69, 1935{\natexlab{a}}.

\bibitem[Schneider(1935{\natexlab{b}})]{schneiderb}
T.~Schneider.
\newblock Transzendenzuntersuchungen periodischer {F}unktionen {II}.
  {T}ranszendenzeigenschaften elliptischer {F}unktionen.
\newblock \emph{Journal f{\"u}r die reine und angewandte Mathematik},
  172:\penalty0 70--74, 1935{\natexlab{b}}.

\bibitem[Tao(2008)]{Tao08}
T.~Tao.
\newblock \emph{Structure and randomness: pages from year one of a mathematical
  blog}.
\newblock American Mathematical Society, 2008.

\bibitem[Tarski(1951)]{Tarski51}
A.~Tarski.
\newblock A decision method for elementary algebra and geometry.
\newblock 1951.

\bibitem[Wilkie(1996)]{wilkieMC}
A.~J. Wilkie.
\newblock Model completeness results for expansions of the ordered field of
  real numbers by restricted {P}faffian functions and the exponential function.
\newblock \emph{Journal of the American Mathematical Society}, pages
  1051--1094, 1996.

\end{thebibliography}

\appendix

\section{Proofs of Hardness Lemmas}\label{sec:hardnessProofs}

Throughout this section, let
\[ 
f(t) \defn e^t(1-\cos(t)) + t (1 - \cos(at)) - c|\sin(at)|.
\]

\setcounter{venchoRestateLemma}{\arabic{theorem}}
\setcounter{theorem}{\arabic{venchoRestateLemma1}}

\begin{lemma}
Fix $a\in\ra$ and $\varepsilon,c\in\rats$ with $a,c>0$
and $\varepsilon\in(0, 1)$. If $f(t) = 0$ for infinitely many
$t\geq 0$, then $L_{\infty}(a)\leq c/2\pi^2(1-\varepsilon)$.
%There exists an effective threshold $T$,
%dependent on $a, c, \varepsilon$, such that if $f(t)=0$ for some
%$t\geq T$, then $L(a)\leq c/2\pi^2(1-\varepsilon)$.
\end{lemma}
\setcounter{theorem}{\arabic{venchoRestateLemma}}
\begin{proof}
Suppose $f(t) = 0$ for infinitely many $t$.
Clearly, this also entails $f(t)=0$ for infinitely
many $t\geq T$, for any particular threshold $T\geq 0$.
(Indeed, $f(t)=\min\{f_1(t),f_2(t)\}$ for exponential
polynomials $f_1$ and $f_2$ given at the beginning of Section
\ref{sec:hardness}. Thus, on any bounded
interval, $f$ has no more zeros than $f_1$ and
$f_2$ combined, i.e., only finitely many, by the analiticity 
of $f_1$ and $f_2$.)
We will show that $T$ can be chosen in such a way that
every zero of $f(t)$ on $[T,\infty)$ yields a pair
$(n,m)\in\nat^2$ which satisfies the inequality
\[ 
\left| a - \frac{n}{m} \right| < \frac{c}{2\pi^2m^2(1-\varepsilon)}.
\]
This is sufficient, since infinitely many zeros of $f$ yield
infinitely many solutions, and therefore witness 
$L_{\infty}(a)\leq c/2\pi^2(1-\varepsilon)$.

Thus, consider some $t$ such that $f(t)=0$ and $t\geq T$ for some 
threshold $T$ to be specified later. 
Let $t = 2\pi m + \delta_1$ 
and $at = 2\pi n + \delta_2$, where $m, n \in \nat$ and 
$\delta_1,\delta_2\in [-\pi,\pi)$.  Then we have
\[
\left| a - \frac{n}{m} \right| = 
\frac{|\delta_2-a\delta_1|}{2\pi m}.
\]
We will show that for $T$ large enough, $f(t)=0$ for
$t\geq T$ allows us to 
bound $|\delta_2|$ and $|a\delta_1|$ separately from above
and then apply the triangle inequality to bound $|\delta_2 -a\delta_1|$.

First, choose $\varphi_1,\varphi_2\in(0, 1)$ such that 
$1-\varphi_2 > 1-\varphi_1 > 1-\varepsilon$. Let $T$ be large enough
for the following property to hold:
\[
\frac{t+\pi}{t-2\pi} \leq \frac{1-\varphi_2}{1-\varphi_1}
\mbox{ for all $t\geq T$.}
\]
In particular, since $m= (t-\delta_1)/2\pi$ and $|\delta_1|\leq \pi$, we
have
\begin{equation}\label{hardprop1}
\frac{2m}{2m-1} \leq \frac{t+\pi}{t-2\pi} \leq \frac{1-\varphi_2}{1-\varphi_1}.
\end{equation}
Let also $T$ be large enough to make the following property valid:
\begin{equation}\label{hardprop2}
1-\cos(x)\leq \frac{c|x|}{T} \wedge |x|\leq\pi \Rightarrow
\frac{(1-\varphi_2)x^2}{2} \leq 1-\cos(x).
\end{equation}
Now we have the following chain of inequalities:
\begin{align*}
& 1-\cos(\delta_2) &  \\
\leq &   \mbox{ \{ $f(t) = 0$, noting $e^t(1-\cos(t))\geq 0$ \} } & \\
& \frac{c|\sin(\delta_2)|}{t} & \\
\leq & \mbox{ \{ by $|\sin(x)|\leq|x|$ \} } & \\
& \frac{c|\delta_2|}{t}. & \\
\end{align*}
Then by (\ref{hardprop2}), we have
\[
1-\cos(\delta_2) \geq \frac{(1-\varphi_2)\delta_2^2}{2}.
\]
Thus, combining the upper and lower bounds on $1-\cos(\delta_2)$
and using (\ref{hardprop1}) on the last step, we have
\[
|\delta_2| \leq \frac{2c}{t(1-\varphi_2)} \leq 
\frac{2c}{(2m-1)\pi(1-\varphi_2)} \leq
\frac{c}{m\pi(1-\varphi_1)}.
\]

Second, let $\alpha \defn (1-\varepsilon)^{-1} - (1-\varphi_1)^{-1} > 0$. Let
the threshold $T$ be large enough so that
\begin{equation}\label{hardprop3}
e^{-t} \leq \frac{c\alpha^2}{4\pi^2a^2}\left(\frac{2\pi}{t+\pi}\right)^2 \mbox{ 
for $t\geq T$}
\end{equation}
and
\begin{equation}\label{hardprop4}
%\mbox{if $1-\cos(x)\leq c|x|/e^T$ and $|x|\leq\pi$, then
\mbox{if $1-\cos(x)\leq c/e^T$ and $|x|\leq\pi$, then
$x^2/4 \leq 1-\cos(x)$.}
\end{equation}
The following chain of inequalities holds:
\begin{align*}
& 1-\cos(\delta_1) &  \\
= &   \mbox{ \{ by $f(t) = 0$ \} } & \\
& \frac{c|\sin(\delta_2)| - t(1-\cos(\delta_2))}{e^t} & \\
\leq &   \mbox{ \{ by $|\sin(\delta_2)|, |\cos(\delta_2)| \leq 1$\}} & \\ 
& \frac{c}{e^t} & \\
\leq &  \mbox{ \{ by (\ref{hardprop3}) \}} & \\
& 
\frac{c^2\alpha^2}{4\pi^2a^2}\left(\frac{2\pi}{t+\pi}\right)^2 \\
\leq & \mbox{ \{ by $|\delta_1|\leq \pi$ \}} & \\
& \frac{c^2\alpha^2}{4\pi^2a^2}\left(\frac{2\pi}{t-\delta_1}\right)^2 \\
= & \mbox{ \{ $t=2\pi m + \delta_1$ \}} & \\
& \frac{c^2\alpha^2}{4\pi^2a^2m^2}.
\end{align*}
Moreover, as $1-\cos(\delta_1) \leq ce^{-t} \leq ce^{-T}$,
by (\ref{hardprop4}), we have
\[ 1-\cos(\delta_1) \geq \frac{\delta_1^2}{4}, \]
so combining the lower and upper bound on $1-\cos(\delta_1)$,
we can conclude
\[ |a\delta_1| \leq \frac{c\alpha}{\pi m}. \]

Finally, by the triangle inequality and the bounds on $|a\delta_1|$
and $|\delta_2|$,
we have
\begin{align*}
\left| a - \frac{n}{m} \right| & = \frac{|\delta_2 - a\delta_1|}{2\pi m}
\leq \frac{|\delta_2| + |a\delta_1|}{2\pi m} \\
& \leq \frac{c}{2 \pi^2 m^2}\left(\alpha + \frac{1}{1-\varphi_1}\right) =
\frac{c}{2\pi^2m^2(1-\varepsilon)}.
\end{align*}
Now, by the premise of the Lemma, there are infinitely many $t\geq T$
such that $f(t)=0$, each yielding a pair $(n,m)\in\nat^2$ which satisfies 
the above inequality. These infinitely many pairs $(n,m)$ witness
$L_{\infty}(a)\leq c/2\pi^2(1-\varepsilon)$, as required.
\end{proof}

\setcounter{venchoRestateLemma}{\arabic{theorem}}
\setcounter{theorem}{\arabic{venchoRestateLemma2}}
\begin{lemma}
Fix $a\in \ra$ and $\varepsilon,c\in\rats$ with $a,c>0$ and
$\varepsilon\in(0, 1)$. 
If $L_{\infty}(a)\leq c(1-\varepsilon)/2\pi^2$, then $f(t)=0$
for infinitely many $t$.
\end{lemma}
\setcounter{theorem}{\arabic{venchoRestateLemma}}
\begin{proof}
We will show that there exists an effective threshold $M$,
dependent on $a, c, \varepsilon$, such that if
\begin{equation}\label{eq:lagr} 
\left|a - \frac{n}{m}\right| \leq \frac{c(1-\varepsilon)}{2\pi^2m^2} 
\end{equation}
for natural numbers $n, m$ with $m \geq M$, then $f(2\pi m) \leq 0$. 
Note that this is sufficient to prove the Lemma: the premise guarantees infinitely
many solutions $(n,m)\in\nats^2$ of (\ref{eq:lagr}), so there must
be infinitely many solutions with $m\geq M$, each yielding
$f(2\pi m)\leq 0$. Since $f$ is continuous and moreover is positive for arbitrarily
large times, it must have infinitely many zeros on $[2\pi M,\infty)$.

Now let $M$ be large enough, so that $c(1-\varepsilon)/\pi M < \pi$ and
\begin{equation}\label{hardprop5}
\mbox{if $|x|< c(1-\varepsilon)/\pi M$, then 
$(1-\varepsilon)|x|\leq|\sin(x)|$.}
\end{equation}
Suppose that (\ref{eq:lagr}) holds for $n, m \in\nats$ with $m\geq M$ and 
write $t \defn 2\pi m$. We will show that $f(t) \leq 0$. 
By (\ref{eq:lagr}), we have $|am-n| \leq c(1-\varepsilon)/2\pi^2m$. Therefore,
$at = 2\pi a m = 2\pi n + \delta$ where 
$|\delta| \leq c(1-\varepsilon)/\pi m < \pi$. We have
\begin{align*}
& f(t) & \\
= & \mbox{ \{ as $\cos(t)=1$ \} } & \\
& t(1-\cos(\delta)) - c|\sin(\delta)| & \\
\leq & \mbox{ \{ by (\ref{hardprop5}) and $1-\cos(x)\leq x^2/2$ \} } & \\
& \pi m\delta^2 - c(1-\varepsilon)|\delta| & \\
\leq & \mbox{ \{ by $|\delta|\leq c(1-\varepsilon)/\pi m$ \} } & \\
& 0. &
\end{align*}
\end{proof}

%The following corollary is immediate:
%\begin{lemma}\label{hardBackward}
%Fix $a\in \ra$ and $\varepsilon,c\in\rats$ with $a,c>0$ and
%$\varepsilon\in(0, 1)$. There exists an effective threshold $T$,
%dependent on $a, c, \varepsilon$,
%such that if $f(t)\neq 0$ for all $t\geq T$, then either
%$L(a) < c(1-\varepsilon)/2\pi^2$ and this is witnessed by
%natural numbers $n, m$ with $m < T/2\pi$, or 
%$L(a)\geq c(1-\varepsilon)/2\pi^2$.
%\end{lemma}

% The bibliography should be embedded for final submission.
%\begin{thebibliography}{}
%\softraggedright
%\bibitem[Smith et~al.(2009)Smith, Jones]{smith02}
%P. Q. Smith, and X. Y. Jones. ...reference text...
%\end{thebibliography}

\end{document}